\documentclass[11pt, a4paper,onecolumn]{article}
\pdfoutput=1
\usepackage[utf8]{inputenc}
\usepackage{amsmath}
\usepackage{amssymb}
\usepackage{stmaryrd}
\usepackage{graphicx}
\usepackage[dvipsnames]{xcolor}
\usepackage{multirow}
\usepackage[normalem]{ulem}
\usepackage{comment}
\usepackage{authblk}

\usepackage{compactbib}
\usepackage{paralist}

\newcommand{\C}{\mathbb{C}}
\newcommand{\R}{\mathbb{R}}
\newcommand{\cH}{\mathcal{H}}
\newcommand{\cS}{\mathcal{S}}
\newcommand{\cE}{\mathcal{E}}
\newcommand{\cT}{\mathcal{T}}
\newcommand{\bo}{\boldsymbol{\omega}}
\newcommand{\be}{\boldsymbol{e}}

\newcommand{\I}{\operatorname{I}}
\newcommand{\tr}{\operatorname{tr}}
\def\ket#1{|#1\rangle}

\def\ketbra#1#2{|#1\rangle\langle#2|}
\def\ketbraq#1{|#1\rangle\langle#1|}

\newcommand{\Ks}{$K_\mathsf{Sym}[\mathcal{P}]$ }

\DeclareMathOperator{\Conv}{Conv}

\newcommand{\jnote}[1]{\textcolor{Red}{[JHS:#1]}}
\newcommand{\bel}{\color{magenta}}

\newcommand{\blk}{\color{black}}
\newcommand{\bnote}[1]{\textcolor{orange}{#1}}

\usepackage{dsfont}
\usepackage{rotating}
\usepackage{floatpag}
\rotfloatpagestyle{empty}

\usepackage{amsmath,amsthm,amssymb}
\usepackage{xspace,enumerate,color,epsfig}
\usepackage{graphicx}
\graphicspath{{.}{./figures/}}
\usepackage{marvosym}

\usepackage{tikzfig}
\usepackage{stmaryrd}
\usepackage{docmute}
\usepackage{keycommand}

\usepackage{pifont}

\usepackage{enumitem}

\newkeycommand{\pointmap}[style=point][1]{\,\tikz{\node[style=\commandkey{style}] (x) at (0,-0.3) {$#1$};\node [style=none] (3) at (0, 1.0) {};\node [style=none] (3) at (0, -1.0) {};\draw (x)--(0,0.7);}\,}

\newkeycommand{\typedkpoint}[style=kpoint,edgestyle=][2]{%
\,\begin{tikzpicture}
  \begin{pgfonlayer}{nodelayer}
    \node [style=none] (0) at (0, 1) {};
    \node [style=\commandkey{style}] (1) at (0, -0.75) {$#2$};
    \node [style=right label] (2) at (0.25, 0.5) {$#1$};
  \end{pgfonlayer}
  \begin{pgfonlayer}{edgelayer}
    \draw [\commandkey{edgestyle}] (1) to (0.center);
  \end{pgfonlayer}
\end{tikzpicture}\,}

\newkeycommand{\typedkpointdag}[style=kpoint adjoint,edgestyle=][2]{%
\,\begin{tikzpicture}
  \begin{pgfonlayer}{nodelayer}
    \node [style=none] (0) at (0, -1) {};
    \node [style=\commandkey{style}] (1) at (0, 0.75) {$#2$};
    \node [style=right label] (2) at (0.25, -0.5) {$#1$};
  \end{pgfonlayer}
  \begin{pgfonlayer}{edgelayer}
    \draw [\commandkey{edgestyle}] (0.center) to (1);
  \end{pgfonlayer}
\end{tikzpicture}\,}

\newcommand{\bistateadj}[1]{\,\begin{tikzpicture}[yshift=-3mm]
  \begin{pgfonlayer}{nodelayer}
    \node [style=kpointadj, minimum width=1 cm, inner sep=2pt] (0) at (0, 0.5) {$\psi$};
    \node [style=none] (1) at (-0.75, 0.25) {};
    \node [style=none] (2) at (0.75, 0.25) {};
    \node [style=none] (3) at (-0.75, -0.5) {};
    \node [style=none] (4) at (0.75, -0.5) {};
  \end{pgfonlayer}
  \begin{pgfonlayer}{edgelayer}
    \draw (1.center) to (3.center);
    \draw (2.center) to (4.center);
  \end{pgfonlayer}
\end{tikzpicture}\,}

\newkeycommand{\pointketbra}[style1=copoint,style2=point][2]{\,%
\begin{tikzpicture}
  \begin{pgfonlayer}{nodelayer}
    \node [style=none] (0) at (0, -1.5) {};
    \node [style=\commandkey{style1}] (1) at (0, -0.75) {$#1$};
    \node [style=\commandkey{style2}] (2) at (0, 0.75) {$#2$};
    \node [style=none] (3) at (0, 1.5) {};
  \end{pgfonlayer}
  \begin{pgfonlayer}{edgelayer}
    \draw (0.center) to (1);
    \draw (2) to (3.center);
  \end{pgfonlayer}
\end{tikzpicture}\,}

\newkeycommand{\twocopointketbra}[style1=copoint,style2=copoint,style3=point][3]{\,%
\begin{tikzpicture}
  \begin{pgfonlayer}{nodelayer}
    \node [style=none] (0) at (-0.75, -1.5) {};
    \node [style=none] (1) at (0.75, -1.5) {};
    \node [style=\commandkey{style1}] (2) at (-0.75, -0.75) {$#1$};
    \node [style=\commandkey{style2}] (3) at (0.75, -0.75) {$#2$};
    \node [style=\commandkey{style3}] (4) at (0, 0.75) {$#3$};
    \node [style=none] (5) at (0, 1.5) {};
  \end{pgfonlayer}
  \begin{pgfonlayer}{edgelayer}
    \draw (0.center) to (2);
    \draw (1.center) to (3);
    \draw (4) to (5.center);
  \end{pgfonlayer}
\end{tikzpicture}\,}

\newkeycommand{\twopointketbra}[style1=copoint,style2=point,style3=point][3]{\,%
\begin{tikzpicture}
  \begin{pgfonlayer}{nodelayer}
    \node [style=none] (0) at (0, -1.5) {};
    \node [style=none] (1) at (-0.75, 1.5) {};
    \node [style=\commandkey{style1}] (2) at (0, -0.75) {$#1$};
    \node [style=\commandkey{style2}] (3) at (-0.75, 0.75) {$#2$};
    \node [style=\commandkey{style3}] (4) at (0.75, 0.75) {$#3$};
    \node [style=none] (5) at (0.75, 1.5) {};
  \end{pgfonlayer}
  \begin{pgfonlayer}{edgelayer}
    \draw (0.center) to (2);
    \draw (1.center) to (3);
    \draw (4) to (5.center);
  \end{pgfonlayer}
\end{tikzpicture}\,}

\newkeycommand{\pointbraket}[style1=point,style2=copoint][2]{\,%
\begin{tikzpicture}
  \begin{pgfonlayer}{nodelayer}
    \node [style=\commandkey{style1}] (0) at (0, -0.5) {$#1$};
    \node [style=\commandkey{style2}] (1) at (0, 0.5) {$#2$};
  \end{pgfonlayer}
  \begin{pgfonlayer}{edgelayer}
    \draw (0) to (1);
  \end{pgfonlayer}
\end{tikzpicture}\,}

\newkeycommand{\kpointbraket}[style1=kpoint,style2=kpoint adjoint][2]{\,%
\begin{tikzpicture}
  \begin{pgfonlayer}{nodelayer}
    \node [style=\commandkey{style1}] (0) at (0, -0.75) {$#1$};
    \node [style=\commandkey{style2}] (1) at (0, 0.75) {$#2$};
  \end{pgfonlayer}
  \begin{pgfonlayer}{edgelayer}
    \draw (0) to (1);
  \end{pgfonlayer}
\end{tikzpicture}\,}

\newkeycommand{\kpointmap}[style1=kpoint,style2=map][2]{\,%
\begin{tikzpicture}
  \begin{pgfonlayer}{nodelayer}
    \node [style=\commandkey{style1}] (0) at (0, -1.1) {$#1$};
    \node [style=\commandkey{style2}] (1) at (0, 0.2) {$#2$};
    \node [style=none] (2) at (0, 1.5) {};
  \end{pgfonlayer}
  \begin{pgfonlayer}{edgelayer}
    \draw (0) to (1);
    \draw (1) to (2.center);
  \end{pgfonlayer}
\end{tikzpicture}%
\,}

\newkeycommand{\sandwichtwo}[style1=point,style2=map,style3=map,style4=copoint][4]{\,%
\begin{tikzpicture}
  \begin{pgfonlayer}{nodelayer}
    \node [style=\commandkey{style2}] (0) at (0, -0.75) {$#2$};
    \node [style=\commandkey{style3}] (1) at (0, 0.75) {$#3$};
    \node [style=\commandkey{style1}] (2) at (0, -2) {$#1$};
    \node [style=\commandkey{style4}] (3) at (0, 2) {$#4$};
  \end{pgfonlayer}
  \begin{pgfonlayer}{edgelayer}
    \draw (2) to (0);
    \draw (0) to (1);
    \draw (1) to (3);
  \end{pgfonlayer}
\end{tikzpicture}\,}

\newkeycommand{\longbraket}[style1=point,style2=copoint][2]{\,%
\begin{tikzpicture}
  \begin{pgfonlayer}{nodelayer}
    \node [style=\commandkey{style1}] (0) at (0, -2) {$#1$};
    \node [style=\commandkey{style2}] (1) at (0, 2) {$#2$};
  \end{pgfonlayer}
  \begin{pgfonlayer}{edgelayer}
    \draw (0) to (1);
  \end{pgfonlayer}
\end{tikzpicture}\,}

\newkeycommand{\onb}[style=point]{\ensuremath{\left\{\tikz{\node[style=\commandkey{style}] (x) at (0,-0.3) {$j$};\draw (x)--(0,0.7);}\right\}}\xspace}

\newkeycommand{\copointmap}[style=copoint][1]{\,\tikz{\node[style=\commandkey{style}] (x) at (0,0.3) {$#1$};\draw (x)--(0,-0.7);}\,}

\tikzstyle{cdiag}=[matrix of math nodes, row sep=3em, column sep=3em, text height=1.5ex, text depth=0.25ex,inner sep=0.5em]
\tikzstyle{arrow above}=[transform canvas={yshift=0.5ex}]
\tikzstyle{arrow below}=[transform canvas={yshift=-0.5ex}]

\usetikzlibrary{shapes.multipart}

\tikzstyle{CqWire}=[color=gray,line width = .75pt,->-]
\tikzstyle{CcWire}=[cWire]
\tikzstyle{RqWire}=[line width = 1pt, color=black,-<-]
\tikzstyle{RcWire}=[cWire]
\tikzstyle{env}=[copoint,regular polygon rotate=0,minimum width=0.2cm, fill=black]

\tikzstyle{probs}=[shape=semicircle,fill=white,draw=black,shape border rotate=180,minimum width=1.2cm]

\tikzset{->-/.style={decoration={
  markings,
  mark=at position .5 with {\arrow{>}}},postaction={decorate}}}
\tikzset{-<-/.style={decoration={
  markings,
  mark=at position .5 with {\arrow{<}}},postaction={decorate}}}

\tikzstyle{bwSpider}=[
       rectangle split,
       rectangle split parts=2,
       rectangle split part fill={black,white},
 minimum size=3.6 mm, inner sep=-2mm, draw=black,scale=0.5,rounded corners=0.8 mm
       ]
 \tikzstyle{wbSpider}=[
       rectangle split,
       rectangle split parts=2,
       rectangle split part fill={white,black},
 minimum size=3.6 mm, inner sep=-2mm, draw=black,scale=0.5,rounded corners=0.8 mm
       ]

\tikzstyle{every picture}=[baseline=-0.25em,scale=0.5]
\tikzstyle{dotpic}=[] 
\tikzstyle{diredges}=[every to/.style={diredge}]
\tikzstyle{math matrix}=[matrix of math nodes,left delimiter=(,right delimiter=),inner sep=2pt,column sep=1em,row sep=0.5em,nodes={inner sep=0pt},text height=1.5ex, text depth=0.25ex]

\tikzstyle{inline text}=[text height=1.5ex, text depth=0.25ex,yshift=0.5mm]
\tikzstyle{label}=[font=\footnotesize,text height=1.5ex, text depth=0.25ex,yshift=0.5mm]
\tikzstyle{left label}=[label,anchor=east,xshift=1.5mm]
\tikzstyle{right label}=[label,anchor=west,xshift=-1mm]
\tikzstyle{up label}=[label,anchor=south,yshift=-1mm]

\tikzstyle{braceedge}=[decorate,decoration={brace,amplitude=2mm,raise=-1mm}]
\tikzstyle{small braceedge}=[decorate,decoration={brace,amplitude=1mm,raise=-1mm}]

\tikzstyle{doubled}=[line width=1.6pt] 
\tikzstyle{boldedge}=[doubled,shorten <=-0.17mm,shorten >=-0.17mm]
\tikzstyle{boldedgegray}=[doubled,gray,shorten <=-0.17mm,shorten >=-0.17mm]
\tikzstyle{singleedgegray}=[gray]

\tikzstyle{semidoubled}=[line width=1.4pt] 
\tikzstyle{semiboldedgegray}=[semidoubled,gray,shorten <=-0.17mm,shorten >=-0.17mm]

\tikzstyle{boxedge}=[semiboldedgegray]

\tikzstyle{boldedgedashed}=[very thick,dashed,shorten <=-0.17mm,shorten >=-0.17mm]
\tikzstyle{vboldedgedashed}=[doubled,dashed,shorten <=-0.17mm,shorten >=-0.17mm]
\tikzstyle{left hook arrow}=[left hook-latex]
\tikzstyle{right hook arrow}=[right hook-latex]
\tikzstyle{sembracket}=[line width=0.5pt,shorten <=-0.07mm,shorten >=-0.07mm]

\tikzstyle{causal edge}=[->,thick,gray]
\tikzstyle{causal nondir}=[thick,gray]
\tikzstyle{timeline}=[thick,gray, dashed]

\tikzstyle{cedge}=[<->,thick,gray!70!white]

\tikzstyle{empty diagram}=[draw=gray!40!white,dashed,shape=rectangle,minimum width=1cm,minimum height=1cm]
\tikzstyle{empty diagram small}=[draw=gray!50!white,dashed,shape=rectangle,minimum width=0.6cm,minimum height=0.5cm]

\tikzstyle{dot}=[inner sep=0mm,minimum width=2mm,minimum height=2mm,draw,shape=circle]
\tikzstyle{bigdot}=[inner sep=0mm,minimum width=5mm,minimum height=5mm,draw,shape=circle]
\tikzstyle{leak}=[white dot, shape=regular polygon, minimum size=3.3 mm, regular polygon sides=3, outer sep=-0.2mm, regular polygon rotate=270]
\tikzstyle{proj}=[regular polygon,regular polygon sides=4,draw,scale=0.75,inner sep=-0.5pt,minimum width=6mm,fill=white]
\tikzstyle{projOut}=[regular polygon,regular polygon sides=3,draw,scale=0.75,inner sep=-0.5pt,minimum width=7.5mm,fill=white,regular polygon rotate=180]
\tikzstyle{projIn}=[regular polygon,regular polygon sides=3,draw,scale=0.75,inner sep=-0.5pt,minimum width=7.5mm,fill=white]
\tikzstyle{Vleak}=[white dot, shape=regular polygon, minimum size=3.3 mm, regular polygon sides=3, outer sep=-0.2mm, regular polygon rotate=90]
\tikzstyle{dleak}=[white dot, line width=1.6pt, shape=regular polygon, minimum size=3.3 mm, regular polygon sides=3, outer sep=-0.2mm, regular polygon rotate=270]

\tikzstyle{Wsquare}=[white dot, shape=regular polygon, rounded corners=0.8 mm, minimum size=3.3 mm, regular polygon sides=3, outer sep=-0.2mm]
\tikzstyle{Wsquareadj}=[white dot, shape=regular polygon, rounded corners=0.8 mm, minimum size=3.3 mm, regular polygon sides=3, outer sep=-0.2mm, regular polygon rotate=180]
\tikzstyle{ddot}=[inner sep=0mm, doubled, minimum width=2.5mm,minimum height=2.5mm,draw,shape=circle]

\tikzstyle{black dot}=[dot,fill=black]
\tikzstyle{white Wsquare}=[Wsquare,fill=gray,text depth=-0.2mm]
\tikzstyle{white Wsquareadj}=[Wsquareadj,fill=white,text depth=-0.2mm]
\tikzstyle{green dot}=[white dot] 
\tikzstyle{gray dot}=[dot,fill=gray,text depth=-0.2mm]
\tikzstyle{red dot}=[gray dot] 

\tikzstyle{black ddot}=[ddot,fill=black]
\tikzstyle{white ddot}=[ddot,fill=white]
\tikzstyle{gray ddot}=[ddot,fill=gray!40!white]

\tikzstyle{gray edge}=[gray!60!white]

\tikzstyle{small dot}=[inner sep=0.2mm,minimum width=0pt,minimum height=0pt,draw,shape=circle]

\tikzstyle{small black dot}=[small dot,fill=black]
\tikzstyle{small white dot}=[small dot,fill=white]
\tikzstyle{small gray dot}=[small dot,fill=gray,draw=gray]

\tikzstyle{causal dot}=[inner sep=0.4mm,minimum width=0pt,minimum height=0pt,draw=white,shape=circle,fill=gray!40!white]

\tikzstyle{phase dimensions}=[minimum size=5mm,font=\footnotesize,rectangle,rounded corners=2.5mm,inner sep=0.2mm,outer sep=-2mm]
\tikzstyle{dphase dimensions}=[minimum size=5mm,font=\footnotesize,rectangle,rounded corners=2.5mm,inner sep=0.2mm,outer sep=-2mm]

\tikzstyle{white phase dot}=[dot,fill=white,phase dimensions]
\tikzstyle{white phase ddot}=[ddot,fill=white,dphase dimensions]

\tikzstyle{white rect ddot}=[draw=black,fill=white,doubled,minimum size=5mm,font=\footnotesize,rectangle,rounded corners=2.5mm,inner sep=0.2mm]
\tikzstyle{gray rect ddot}=[draw=black,fill=gray!40!white,doubled,minimum size=6mm,font=\footnotesize,rectangle,rounded corners=3mm]

\tikzstyle{gray phase dot}=[dot,fill=gray!40!white,phase dimensions]
\tikzstyle{gray phase ddot}=[ddot,fill=gray!40!white,dphase dimensions]
\tikzstyle{grey phase dot}=[gray phase dot]
\tikzstyle{grey phase ddot}=[gray phase ddot]

\tikzstyle{small phase dimensions}=[minimum size=4mm,font=\tiny,rectangle,rounded corners=2mm,inner sep=0.2mm,outer sep=-2mm]
\tikzstyle{small dphase dimensions}=[minimum size=4mm,font=\tiny,rectangle,rounded corners=2mm,inner sep=0.2mm,outer sep=-2mm]

\tikzstyle{small gray phase dot}=[dot,fill=gray!40!white,small phase dimensions]
\tikzstyle{small gray phase ddot}=[ddot,fill=gray!40!white,small dphase dimensions]

\tikzstyle{small map}=[draw,shape=rectangle,minimum height=4mm,minimum width=4mm,fill=white]

\tikzstyle{cnot}=[fill=white,shape=circle,inner sep=-1.4pt]

\tikzstyle{asym hadamard}=[fill=white,draw,shape=NEbox,inner sep=0.6mm,font=\footnotesize,minimum height=4mm]
\tikzstyle{asym hadamard conj}=[fill=white,draw,shape=NWbox,inner sep=0.6mm,font=\footnotesize,minimum height=4mm]
\tikzstyle{asym hadamard dag}=[fill=white,draw,shape=SEbox,inner sep=0.6mm,font=\footnotesize,minimum height=4mm]

\tikzstyle{hadamard}=[fill=white,draw,inner sep=0.6mm,font=\footnotesize,minimum height=4mm,minimum width=4mm]
\tikzstyle{small hadamard}=[fill=white,draw,inner sep=0.6mm,minimum height=1.5mm,minimum width=1.5mm]
\tikzstyle{small hadamard rotate}=[small hadamard,rotate=45]
\tikzstyle{dhadamard}=[hadamard,doubled]
\tikzstyle{small dhadamard}=[small hadamard,doubled]
\tikzstyle{small dhadamard rotate}=[small hadamard rotate,doubled]
\tikzstyle{antipode}=[white dot,inner sep=0.3mm,font=\footnotesize]

\tikzstyle{small gray box}=[small box,fill=gray!30]
\tikzstyle{medium box}=[rectangle,inline text,fill=white,draw,minimum height=5mm,yshift=-0.5mm,minimum width=10mm,font=\small]
\tikzstyle{square box}=[small box] 
\tikzstyle{medium gray box}=[small box,fill=gray!30]
\tikzstyle{semilarge box}=[rectangle,inline text,fill=white,draw,minimum height=5mm,yshift=-0.5mm,minimum width=12.5mm,font=\small]
\tikzstyle{large box}=[rectangle,inline text,fill=white,draw,minimum height=5mm,yshift=-0.5mm,minimum width=15mm,font=\small]
\tikzstyle{large gray box}=[small box,fill=gray!30]

\tikzstyle{Bayes box}=[rectangle,fill=black,draw, minimum height=3mm, minimum width=3mm]

\tikzstyle{gray square point}=[small box,fill=gray!50]

\tikzstyle{dphase box white}=[dhadamard]
\tikzstyle{dphase box gray}=[dhadamard,fill=gray!50!white]
\tikzstyle{phase box white}=[hadamard]
\tikzstyle{phase box gray}=[hadamard,fill=gray!50!white]

\tikzstyle{point nosep}=[regular polygon,regular polygon sides=3,draw,scale=0.75,inner sep=-2pt,minimum width=9mm,fill=white,regular polygon rotate=180]
\tikzstyle{dpoint}=[point,doubled]
\tikzstyle{dcopoint}=[copoint,doubled]

\tikzstyle{pointgrow}=[shape=cornerpoint,kpoint common,scale=0.75,inner sep=3pt]
\tikzstyle{pointgrow dag}=[shape=cornercopoint,kpoint common,scale=0.75,inner sep=3pt]

\tikzstyle{wide copoint}=[fill=white,draw,shape=isosceles triangle,shape border rotate=90,isosceles triangle stretches=true,inner sep=0pt,minimum width=1.5cm,minimum height=6.12mm]
\tikzstyle{wide point}=[fill=white,draw,shape=isosceles triangle,shape border rotate=-90,isosceles triangle stretches=true,inner sep=0pt,minimum width=1.5cm,minimum height=6.12mm,yshift=-0.0mm]
\tikzstyle{wide point plus}=[fill=white,draw,shape=isosceles triangle,shape border rotate=-90,isosceles triangle stretches=true,inner sep=0pt,minimum width=1.74cm,minimum height=7mm,yshift=-0.0mm]

\tikzstyle{wide dpoint}=[fill=white,doubled,draw,shape=isosceles triangle,shape border rotate=-90,isosceles triangle stretches=true,inner sep=0pt,minimum width=1.5cm,minimum height=6.12mm,yshift=-0.0mm]

\tikzstyle{tinypoint}=[regular polygon,regular polygon sides=3,draw,scale=0.55,inner sep=-0.15pt,minimum width=6mm,fill=white,regular polygon rotate=180]

\tikzstyle{white point}=[point]
\tikzstyle{white dpoint}=[dpoint]
\tikzstyle{green point}=[white point] 
\tikzstyle{white copoint}=[copoint]
\tikzstyle{gray point}=[point,fill=gray!40!white]
\tikzstyle{gray dpoint}=[gray point,doubled]
\tikzstyle{red point}=[gray point] 
\tikzstyle{gray copoint}=[copoint,fill=gray!40!white]
\tikzstyle{gray dcopoint}=[gray copoint,doubled]

\tikzstyle{white point guide}=[regular polygon,regular polygon sides=3,font=\scriptsize,draw,scale=0.65,inner sep=-0.5pt,minimum width=9mm,fill=white,regular polygon rotate=180]

\tikzstyle{black point}=[point,fill=black,font=\color{white}]
\tikzstyle{black copoint}=[copoint,fill=black,font=\color{white}]

\tikzstyle{tiny gray point}=[tinypoint,fill=gray!40!white]

\tikzstyle{diredge}=[->]
\tikzstyle{ddiredge}=[<->]
\tikzstyle{rdiredge}=[<-]
\tikzstyle{thickdiredge}=[->, very thick]
\tikzstyle{pointer edge}=[->,very thick,gray]
\tikzstyle{pointer edge part}=[very thick,gray]
\tikzstyle{dashed edge}=[dashed]
\tikzstyle{thick dashed edge}=[very thick,dashed]
\tikzstyle{thick map edge}=[very thick,|->]

\makeatletter
\newcommand{\boxshape}[3]{%
\pgfdeclareshape{#1}{
\inheritsavedanchors[from=rectangle] 
\inheritanchorborder[from=rectangle]
\inheritanchor[from=rectangle]{center}
\inheritanchor[from=rectangle]{north}
\inheritanchor[from=rectangle]{south}
\inheritanchor[from=rectangle]{west}
\inheritanchor[from=rectangle]{east}
\backgroundpath{
\southwest \pgf@xa=\pgf@x \pgf@ya=\pgf@y
\northeast \pgf@xb=\pgf@x \pgf@yb=\pgf@y

\@tempdima=#2
\@tempdimb=#3

\pgfpathmoveto{\pgfpoint{\pgf@xa - 5pt + \@tempdima}{\pgf@ya}}
\pgfpathlineto{\pgfpoint{\pgf@xa - 5pt - \@tempdima}{\pgf@yb}}
\pgfpathlineto{\pgfpoint{\pgf@xb + 5pt + \@tempdimb}{\pgf@yb}}
\pgfpathlineto{\pgfpoint{\pgf@xb + 5pt - \@tempdimb}{\pgf@ya}}
\pgfpathlineto{\pgfpoint{\pgf@xa - 5pt + \@tempdima}{\pgf@ya}}
\pgfpathclose
}
}}
\newcommand{\smallboxshape}[3]{%
\pgfdeclareshape{#1}{
\inheritsavedanchors[from=rectangle] 
\inheritanchorborder[from=rectangle]
\inheritanchor[from=rectangle]{center}
\inheritanchor[from=rectangle]{north}
\inheritanchor[from=rectangle]{south}
\inheritanchor[from=rectangle]{west}
\inheritanchor[from=rectangle]{east}
\backgroundpath{
\southwest \pgf@xa=\pgf@x \pgf@ya=\pgf@y
\northeast \pgf@xb=\pgf@x \pgf@yb=\pgf@y

\@tempdima=#2
\@tempdimb=#3

\pgfpathmoveto{\pgfpoint{\pgf@xa - 3pt + \@tempdima}{\pgf@ya}}
\pgfpathlineto{\pgfpoint{\pgf@xa - 3pt - \@tempdima}{\pgf@yb}}
\pgfpathlineto{\pgfpoint{\pgf@xb + 3pt + \@tempdimb}{\pgf@yb}}
\pgfpathlineto{\pgfpoint{\pgf@xb + 3pt - \@tempdimb}{\pgf@ya}}
\pgfpathlineto{\pgfpoint{\pgf@xa - 3pt + \@tempdima}{\pgf@ya}}
\pgfpathclose
}
}}

\boxshape{NEbox}{0pt}{3pt}
\boxshape{SEbox}{0pt}{-3pt}
\boxshape{NWbox}{5pt}{0pt}
\boxshape{SWbox}{-5pt}{0pt}
\smallboxshape{SWboxSmall}{-3pt}{0pt}
\boxshape{EBox}{-3pt}{3pt}
\boxshape{WBox}{3pt}{-3pt}
\makeatother

\tikzstyle{cloud}=[shape=cloud,draw,minimum width=1.5cm,minimum height=1.5cm]

\tikzstyle{map}=[draw,shape=NEbox,inner sep=1pt,minimum height=4mm,fill=white]
\tikzstyle{dashedmap}=[draw,dashed,shape=NEbox,inner sep=2pt,minimum height=6mm,fill=white]
\tikzstyle{mapdag}=[draw,shape=SEbox,inner sep=1pt,minimum height=4mm,fill=white]
\tikzstyle{mapadj}=[draw,shape=SEbox,inner sep=2pt,minimum height=6mm,fill=white]
\tikzstyle{maptrans}=[draw,shape=SWbox,inner sep=2pt,minimum height=6mm,fill=white]
\tikzstyle{mapconj}=[draw,shape=NWbox,inner sep=2pt,minimum height=6mm,fill=white]

\tikzstyle{medium map}=[draw,shape=NEbox,inner sep=2pt,minimum height=6mm,fill=white,minimum width=7mm]
\tikzstyle{medium map dag}=[draw,shape=SEbox,inner sep=2pt,minimum height=6mm,fill=white,minimum width=7mm]
\tikzstyle{medium map adj}=[draw,shape=SEbox,inner sep=2pt,minimum height=6mm,fill=white,minimum width=7mm]
\tikzstyle{medium map trans}=[draw,shape=SWbox,inner sep=2pt,minimum height=6mm,fill=white,minimum width=7mm]
\tikzstyle{medium map conj}=[draw,shape=NWbox,inner sep=2pt,minimum height=6mm,fill=white,minimum width=7mm]
\tikzstyle{semilarge map}=[draw,shape=NEbox,inner sep=2pt,minimum height=6mm,fill=white,minimum width=9.5mm]
\tikzstyle{semilarge map trans}=[draw,shape=SWbox,inner sep=2pt,minimum height=6mm,fill=white,minimum width=9.5mm]
\tikzstyle{semilarge map adj}=[draw,shape=SEbox,inner sep=2pt,minimum height=6mm,fill=white,minimum width=9.5mm]
\tikzstyle{semilarge map dag}=[draw,shape=SEbox,inner sep=2pt,minimum height=6mm,fill=white,minimum width=9.5mm]
\tikzstyle{semilarge map conj}=[draw,shape=NWbox,inner sep=2pt,minimum height=6mm,fill=white,minimum width=9.5mm]
\tikzstyle{large map}=[draw,shape=NEbox,inner sep=2pt,minimum height=6mm,fill=white,minimum width=12mm]
\tikzstyle{large map conj}=[draw,shape=NWbox,inner sep=2pt,minimum height=6mm,fill=white,minimum width=12mm]
\tikzstyle{very large map}=[draw,shape=NEbox,inner sep=2pt,minimum height=6mm,fill=white,minimum width=17mm]

\tikzstyle{medium dmap}=[draw,doubled,shape=NEbox,inner sep=2pt,minimum height=6mm,fill=white,minimum width=7mm]
\tikzstyle{medium dmap dag}=[draw,doubled,shape=SEbox,inner sep=2pt,minimum height=6mm,fill=white,minimum width=7mm]
\tikzstyle{medium dmap adj}=[draw,doubled,shape=SEbox,inner sep=2pt,minimum height=6mm,fill=white,minimum width=7mm]
\tikzstyle{medium dmap trans}=[draw,doubled,shape=SWbox,inner sep=2pt,minimum height=6mm,fill=white,minimum width=7mm]
\tikzstyle{medium dmap conj}=[draw,doubled,shape=NWbox,inner sep=2pt,minimum height=6mm,fill=white,minimum width=7mm]
\tikzstyle{semilarge dmap}=[draw,doubled,shape=NEbox,inner sep=2pt,minimum height=6mm,fill=white,minimum width=9.5mm]
\tikzstyle{semilarge dmap trans}=[draw,doubled,shape=SWbox,inner sep=2pt,minimum height=6mm,fill=white,minimum width=9.5mm]
\tikzstyle{semilarge dmap adj}=[draw,doubled,shape=SEbox,inner sep=2pt,minimum height=6mm,fill=white,minimum width=9.5mm]
\tikzstyle{semilarge dmap dag}=[draw,doubled,shape=SEbox,inner sep=2pt,minimum height=6mm,fill=white,minimum width=9.5mm]
\tikzstyle{semilarge dmap conj}=[draw,doubled,shape=NWbox,inner sep=2pt,minimum height=6mm,fill=white,minimum width=9.5mm]
\tikzstyle{large dmap}=[draw,doubled,shape=NEbox,inner sep=2pt,minimum height=6mm,fill=white,minimum width=12mm]
\tikzstyle{large dmap conj}=[draw,doubled,shape=NWbox,inner sep=2pt,minimum height=6mm,fill=white,minimum width=12mm]
\tikzstyle{large dmap trans}=[draw,doubled,shape=SWbox,inner sep=2pt,minimum height=6mm,fill=white,minimum width=12mm]
\tikzstyle{large dmap adj}=[draw,doubled,shape=SEbox,inner sep=2pt,minimum height=6mm,fill=white,minimum width=12mm]
\tikzstyle{large dmap dag}=[draw,doubled,shape=SEbox,inner sep=2pt,minimum height=6mm,fill=white,minimum width=12mm]
\tikzstyle{very large dmap}=[draw,doubled,shape=NEbox,inner sep=2pt,minimum height=6mm,fill=white,minimum width=19.5mm]

\tikzstyle{muxbox}=[draw,shape=rectangle,minimum height=3mm,minimum width=3mm,fill=white]
\tikzstyle{dmuxbox}=[muxbox,doubled]

\tikzstyle{box}=[draw,shape=rectangle,inner sep=2pt,minimum height=6mm,minimum width=6mm,fill=white]
\tikzstyle{dbox}=[draw,doubled,shape=rectangle,inner sep=2pt,minimum height=6mm,minimum width=6mm,fill=white]
\tikzstyle{dmap}=[draw,doubled,shape=NEbox,inner sep=2pt,minimum height=6mm,fill=white]
\tikzstyle{dmapdag}=[draw,doubled,shape=SEbox,inner sep=2pt,minimum height=6mm,fill=white]
\tikzstyle{dmapadj}=[draw,doubled,shape=SEbox,inner sep=2pt,minimum height=6mm,fill=white]
\tikzstyle{dmaptrans}=[draw,doubled,shape=SWbox,inner sep=2pt,minimum height=6mm,fill=white]
\tikzstyle{dmapconj}=[draw,doubled,shape=NWbox,inner sep=2pt,minimum height=6mm,fill=white]

\tikzstyle{ddmap}=[draw,doubled,dashed,shape=NEbox,inner sep=2pt,minimum height=6mm,fill=white]
\tikzstyle{ddmapdag}=[draw,doubled,dashed,shape=SEbox,inner sep=2pt,minimum height=6mm,fill=white]
\tikzstyle{ddmapadj}=[draw,doubled,dashed,shape=SEbox,inner sep=2pt,minimum height=6mm,fill=white]
\tikzstyle{ddmaptrans}=[draw,doubled,dashed,shape=SWbox,inner sep=2pt,minimum height=6mm,fill=white]
\tikzstyle{ddmapconj}=[draw,doubled,dashed,shape=NWbox,inner sep=2pt,minimum height=6mm,fill=white]

\boxshape{sNEbox}{0pt}{3pt}
\boxshape{sSEbox}{0pt}{-3pt}
\boxshape{sNWbox}{3pt}{0pt}
\boxshape{sSWbox}{-3pt}{0pt}
\tikzstyle{smap}=[draw,shape=sNEbox,fill=white]
\tikzstyle{smapdag}=[draw,shape=sSEbox,fill=white]
\tikzstyle{smapadj}=[draw,shape=sSEbox,fill=white]
\tikzstyle{smaptrans}=[draw,shape=sSWbox,fill=white]
\tikzstyle{smapconj}=[draw,shape=sNWbox,fill=white]

\tikzstyle{dsmap}=[draw,dashed,shape=sNEbox,fill=white]
\tikzstyle{dsmapdag}=[draw,dashed,shape=sSEbox,fill=white]
\tikzstyle{dsmaptrans}=[draw,dashed,shape=sSWbox,fill=white]
\tikzstyle{dsmapconj}=[draw,dashed,shape=sNWbox,fill=white]

\boxshape{mNEbox}{0pt}{10pt}
\boxshape{mSEbox}{0pt}{-10pt}
\boxshape{mNWbox}{10pt}{0pt}
\boxshape{mSWbox}{-10pt}{0pt}
\tikzstyle{mmap}=[draw,shape=mNEbox]
\tikzstyle{mmapdag}=[draw,shape=mSEbox]
\tikzstyle{mmaptrans}=[draw,shape=mSWbox]
\tikzstyle{mmapconj}=[draw,shape=mNWbox]

\tikzstyle{mmapgray}=[draw,fill=gray!40!white,shape=mNEbox]
\tikzstyle{smapgray}=[draw,fill=gray!40!white,shape=sNEbox]

\makeatletter

\pgfdeclareshape{cornerpoint}{
\inheritsavedanchors[from=rectangle] 
\inheritanchorborder[from=rectangle]
\inheritanchor[from=rectangle]{center}
\inheritanchor[from=rectangle]{north}
\inheritanchor[from=rectangle]{south}
\inheritanchor[from=rectangle]{west}
\inheritanchor[from=rectangle]{east}
\backgroundpath{
\southwest \pgf@xa=\pgf@x \pgf@ya=\pgf@y
\northeast \pgf@xb=\pgf@x \pgf@yb=\pgf@y

\pgfmathsetmacro{\pgf@shorten@left}{\pgfkeysvalueof{/tikz/shorten left}}
\pgfmathsetmacro{\pgf@shorten@right}{\pgfkeysvalueof{/tikz/shorten right}}

\pgfpathmoveto{\pgfpoint{0.5 * (\pgf@xa + \pgf@xb)}{\pgf@ya - 5pt}}
\pgfpathlineto{\pgfpoint{\pgf@xa - 8pt + \pgf@shorten@left}{\pgf@yb - 1.5 * \pgf@shorten@left}}
\pgfpathlineto{\pgfpoint{\pgf@xa - 8pt + \pgf@shorten@left}{\pgf@yb}}
\pgfpathlineto{\pgfpoint{\pgf@xb + 8pt - \pgf@shorten@right}{\pgf@yb}}
\pgfpathlineto{\pgfpoint{\pgf@xb + 8pt - \pgf@shorten@right}{\pgf@yb - 1.5 * \pgf@shorten@right}}
\pgfpathclose
}
}

\pgfdeclareshape{cornercopoint}{
\inheritsavedanchors[from=rectangle] 
\inheritanchorborder[from=rectangle]
\inheritanchor[from=rectangle]{center}
\inheritanchor[from=rectangle]{north}
\inheritanchor[from=rectangle]{south}
\inheritanchor[from=rectangle]{west}
\inheritanchor[from=rectangle]{east}
\backgroundpath{
\southwest \pgf@xa=\pgf@x \pgf@ya=\pgf@y
\northeast \pgf@xb=\pgf@x \pgf@yb=\pgf@y

\pgfmathsetmacro{\pgf@shorten@left}{\pgfkeysvalueof{/tikz/shorten left}}
\pgfmathsetmacro{\pgf@shorten@right}{\pgfkeysvalueof{/tikz/shorten right}}

\pgfpathmoveto{\pgfpoint{0.5 * (\pgf@xa + \pgf@xb)}{\pgf@yb + 5pt}}
\pgfpathlineto{\pgfpoint{\pgf@xa - 8pt + \pgf@shorten@left}{\pgf@ya + 1.5 * \pgf@shorten@left}}
\pgfpathlineto{\pgfpoint{\pgf@xa - 8pt + \pgf@shorten@left}{\pgf@ya}}
\pgfpathlineto{\pgfpoint{\pgf@xb + 8pt - \pgf@shorten@right}{\pgf@ya}}
\pgfpathlineto{\pgfpoint{\pgf@xb + 8pt - \pgf@shorten@right}{\pgf@ya + 1.5 * \pgf@shorten@right}}
\pgfpathclose
}
}

\makeatother

\pgfkeyssetvalue{/tikz/shorten left}{0pt}
\pgfkeyssetvalue{/tikz/shorten right}{0pt}

\tikzstyle{kpoint common}=[draw,fill=white,inner sep=1pt,minimum height=4mm]
\tikzstyle{kpoint sc}=[shape=cornerpoint,kpoint common]
\tikzstyle{kpoint adjoint sc}=[shape=cornercopoint,kpoint common]
\tikzstyle{kpoint}=[shape=cornerpoint,shorten left=5pt,kpoint common]
\tikzstyle{kpoint adjoint}=[shape=cornercopoint,shorten left=5pt,kpoint common]
\tikzstyle{kpoint conjugate}=[shape=cornerpoint,shorten right=5pt,kpoint common]
\tikzstyle{kpoint transpose}=[shape=cornercopoint,shorten right=5pt,kpoint common]
\tikzstyle{kpoint symm}=[shape=cornerpoint,shorten left=5pt,shorten right=5pt,kpoint common]

\tikzstyle{wide kpoint sc}=[shape=cornerpoint,kpoint common, minimum width=1 cm]
\tikzstyle{wide kpointdag sc}=[shape=cornercopoint,kpoint common, minimum width=1 cm]

\tikzstyle{black kpoint}=[shape=cornerpoint,shorten left=5pt,kpoint common,fill=black,font=\color{white}]

\tikzstyle{black kpoint sm}=[shape=cornerpoint,shorten left=5pt,kpoint common,fill=black,font=\color{white},scale=0.75]

\tikzstyle{black kpoint adjoint}=[shape=cornercopoint,shorten left=5pt,kpoint common,fill=black,font=\color{white}]
\tikzstyle{black kpointadj}=[shape=cornercopoint,shorten left=5pt,kpoint common,fill=black,font=\color{white}]

\tikzstyle{black kpointadj sm}=[shape=cornercopoint,shorten left=5pt,kpoint common,fill=black,font=\color{white},scale=0.75]

\tikzstyle{black dkpoint}=[shape=cornerpoint,shorten left=5pt,kpoint common,fill=black, doubled,font=\color{white}]
\tikzstyle{black dkpoint adjoint}=[shape=cornercopoint,shorten left=5pt,kpoint common,fill=black, doubled,font=\color{white}]
\tikzstyle{black dkpointadj}=[shape=cornercopoint,shorten left=5pt,kpoint common,fill=black, doubled,font=\color{white}]

\tikzstyle{black dkpoint sm}=[shape=cornerpoint,shorten left=5pt,kpoint common,fill=black, doubled,font=\color{white},scale=0.75]
\tikzstyle{black dkpointadj sm}=[shape=cornercopoint,shorten left=5pt,kpoint common,fill=black, doubled,font=\color{white},scale=0.75]

\tikzstyle{kpointdag}=[kpoint adjoint]
\tikzstyle{kpointadj}=[kpoint adjoint]
\tikzstyle{kpointconj}=[kpoint conjugate]
\tikzstyle{kpointtrans}=[kpoint transpose]

\tikzstyle{big kpoint}=[kpoint, minimum width=1.2 cm, minimum height=8mm, inner sep=4pt, text depth=3mm]

\tikzstyle{wide kpoint}=[kpoint, minimum width=1 cm, inner sep=2pt]
\tikzstyle{wide kpointdag}=[kpointdag, minimum width=1 cm, inner sep=2pt]
\tikzstyle{wide kpointconj}=[kpointconj, minimum width=1 cm, inner sep=2pt]
\tikzstyle{wide kpointtrans}=[kpointtrans, minimum width=1 cm, inner sep=2pt]

\tikzstyle{wider kpoint}=[kpoint, minimum width=1.25 cm, inner sep=2pt]
\tikzstyle{wider kpointdag}=[kpointdag, minimum width=1.25 cm, inner sep=2pt]
\tikzstyle{wider kpointconj}=[kpointconj, minimum width=1.25 cm, inner sep=2pt]
\tikzstyle{wider kpointtrans}=[kpointtrans, minimum width=1.25 cm, inner sep=2pt]

\tikzstyle{gray kpoint}=[kpoint,fill=gray!50!white]
\tikzstyle{gray kpointdag}=[kpointdag,fill=gray!50!white]
\tikzstyle{gray kpointadj}=[kpointadj,fill=gray!50!white]
\tikzstyle{gray kpointconj}=[kpointconj,fill=gray!50!white]
\tikzstyle{gray kpointtrans}=[kpointtrans,fill=gray!50!white]

\tikzstyle{gray dkpoint}=[kpoint,fill=gray!50!white,doubled]
\tikzstyle{gray dkpointdag}=[kpointdag,fill=gray!50!white,doubled]
\tikzstyle{gray dkpointadj}=[kpointadj,fill=gray!50!white,doubled]
\tikzstyle{gray dkpointconj}=[kpointconj,fill=gray!50!white,doubled]
\tikzstyle{gray dkpointtrans}=[kpointtrans,fill=gray!50!white,doubled]

\tikzstyle{white label}=[draw,fill=white,rectangle,inner sep=0.7 mm]
\tikzstyle{gray label}=[draw,fill=gray!50!white,rectangle,inner sep=0.7 mm]
\tikzstyle{black label}=[draw,fill=black,rectangle,inner sep=0.7 mm]

\tikzstyle{dkpoint}=[kpoint,doubled]
\tikzstyle{wide dkpoint}=[wide kpoint,doubled]
\tikzstyle{dkpointdag}=[kpoint adjoint,doubled]
\tikzstyle{wide dkpointdag}=[wide kpointdag,doubled]
\tikzstyle{dkcopoint}=[kpoint adjoint,doubled]
\tikzstyle{dkpointadj}=[kpoint adjoint,doubled]
\tikzstyle{dkpointconj}=[kpoint conjugate,doubled]
\tikzstyle{dkpointtrans}=[kpoint transpose,doubled]

\tikzstyle{kscalar}=[kpoint common, shape=EBox, inner xsep=-1pt, inner ysep=3pt,font=\small]
\tikzstyle{kscalarconj}=[kpoint common, shape=WBox, inner xsep=-1pt, inner ysep=3pt,font=\small]

\tikzstyle{spekpoint}=[kpoint sc,minimum height=5mm,inner sep=3pt]
\tikzstyle{spekcopoint}=[kpoint adjoint sc,minimum height=5mm,inner sep=3pt]

\tikzstyle{dspekpoint}=[spekpoint,doubled]
\tikzstyle{dspekcopoint}=[spekcopoint,doubled]

 \tikzstyle{bigground}=[regular polygon,regular polygon sides=3,draw=gray,scale=0.50,inner sep=-0.5pt,minimum width=10mm,fill=gray]

\tikzstyle{arrs}=[-latex,font=\small,auto]
\tikzstyle{arrow plain}=[arrs]
\tikzstyle{arrow dashed}=[dashed,arrs]
\tikzstyle{arrow bold}=[very thick,arrs]
\tikzstyle{arrow hide}=[draw=white!0,-]
\tikzstyle{arrow reverse}=[latex-]
\tikzstyle{cdnode}=[]
\tikzstyle{dscalar}=[diamond,doubled, draw,inner sep=0.5pt,font=\small]

\tikzstyle{epiCopoint}=[regular polygon,regular polygon sides=3,draw,scale=0.75,inner sep=-0.5pt,minimum width=5mm,fill=white,regular polygon rotate=0,line width=1pt]
\tikzstyle{epiPoint}=[regular polygon,regular polygon sides=3,draw,scale=0.75,inner sep=-0.5pt,minimum width=5mm,fill=white,regular polygon rotate=180,line width=1pt]
\tikzstyle{epiPointWide}=[regular polygon,regular polygon sides=3,draw,xscale=0.75,yscale=.5,inner sep=-0.5pt,minimum width=8mm,fill=white,regular polygon rotate=180,line width=1pt]
\tikzstyle{epiBox}=[fill=white,draw, line width = 1pt,inner sep=0.6mm,font=\footnotesize,minimum height=3mm,minimum width=3mm]
\tikzstyle{epiBoxWide}=[fill=white,draw, line width = 1pt,inner sep=0.6mm,font=\footnotesize,minimum height=3mm,minimum width=5mm]
\tikzstyle{epiBoxVeryWide}=[fill=white,draw, line width = 1pt,inner sep=0.6mm,font=\footnotesize,minimum height=3mm,minimum width=7mm]

\tikzstyle{clear dot}=[dot,fill=none,text depth=-0.2mm,draw=gray, line width = .75pt]
\tikzstyle{tall clear dot}=[dot,fill=none,text depth=-0.2mm,draw=gray, line width = .75pt,shape=ellipse, minimum height=5mm]
\tikzstyle{wide clear dot}=[dot,fill=none,text depth=-0.2mm,draw=gray, line width = .75pt, shape=ellipse, minimum width = 5mm]
\tikzstyle{very wide clear dot}=[dot,fill=none,text depth=-0.2mm,draw=gray, line width = .75pt, shape=ellipse, minimum width = 7mm ]

\tikzstyle{point}=[regular polygon,regular polygon sides=3,draw,scale=0.75,inner sep=-0.5pt,minimum width=9mm,fill=white,regular polygon rotate=180]
\tikzstyle{infpoint}=[regular polygon,regular polygon sides=3,draw,scale=0.75,inner sep=-0.5pt,minimum width=9mm,fill=white,regular polygon rotate=90]
\tikzstyle{infcopoint}=[regular polygon,regular polygon sides=3,draw,scale=0.75,inner sep=-0.5pt,minimum width=9mm,fill=white,regular polygon rotate=270]
\tikzstyle{copoint}=[regular polygon,regular polygon sides=3,draw,scale=0.75,inner sep=-0.5pt,minimum width=9mm,fill=white]
\tikzstyle{small box}=[rectangle,inline text,fill=white,draw,minimum height=5mm,yshift=-0.5mm,minimum width=5mm,font=\small]
\tikzstyle{scalar}=[diamond,draw,inner sep=0.5pt,font=\small]

\tikzstyle{white dot}=[dot,fill=white,text depth=-0.2mm]
\tikzstyle{uControl}= [draw, shape=diamond, aspect=.5,inner sep=0pt,minimum height=2mm,minimum width=3.5mm,fill=black]
\tikzstyle{funcApp} = [draw, shape=diamond, aspect=2,inner sep=0pt,minimum height=3.5mm,minimum width=2mm,fill=black]
\tikzstyle{funcSplit} = [draw,shape=regular polygon, regular polygon sides = 3,inner sep=0pt,minimum height=3mm,minimum  width=1mm,regular polygon rotate=210, fill=black]
\tikzstyle{copy} = [black dot]
\tikzstyle{infMerge} = [draw=gray,fill=gray,shape=regular polygon, regular polygon sides = 3,inner sep=0pt,minimum height=2.5mm,minimum  width=1mm,regular polygon rotate=30]
\tikzstyle{infSplit}= [draw=gray,fill=gray,shape=regular polygon, regular polygon sides = 3,inner sep=0pt,minimum height=2.5mm,minimum  width=1mm,regular polygon rotate=210]
\tikzstyle{seqComp}=[white dot]
\tikzstyle{parComp}=[circuit ee IEC, bulb,fill=white]
\tikzstyle{addStar}=[draw, shape=star, star points=5, star rotate=90,minimum size=2mm, inner sep=0pt,fill=black]
\tikzstyle{removeStar}=[draw, shape=star, star points=5, star rotate =270,minimum size=2mm, inner sep=0pt,fill=white]

\tikzstyle{upground}=[circuit ee IEC,thick,ground,rotate=90,xscale=2.5,yscale=2]
 \tikzstyle{downground}=[circuit ee IEC,thick,ground,rotate=-90,xscale=2.5,yscale=2]
 \tikzstyle{infupground}=[circuit ee IEC,thick,ground,rotate=0,xscale=2.5,yscale=2]
  \tikzstyle{ignore}=[circuit ee IEC,thick,ground,rotate=0,xscale=1.5,yscale=1,xshift=5pt]
 \tikzstyle{infdownground}=[circuit ee IEC,thick,ground,rotate=180,xscale=2.5,yscale=2]

\tikzstyle{oWire}=[line width = .75pt, color=green!40!black!70!]
\tikzstyle{qWire}=[line width = 1pt, color=black]
\tikzstyle{cWire}=[color=gray,line width = .75pt]
\tikzstyle{thick gray dashed edge}=[thick dashed edge,gray!40]

\let\olddagger\dagger
\renewcommand{\dagger}{\ensuremath{\olddagger}\xspace}

\usepackage{makeidx}
\makeindex

\theoremstyle{plain}
\newtheorem*{main theorem}{Main Theorem}
\newtheorem{theorem}{Theorem}[section]

\newtheorem{lemma}[theorem]{Lemma}

\newtheorem{definition}[theorem]{Definition}

\newtheorem{example*}[theorem]{Example*}
\newtheorem{examples*}[theorem]{Examples*}

\newtheorem{remark*}[theorem]{Remark*}

\newtheorem*{search problem}{Search Problem}

\usepackage{color}
\def\bR{\begin{color}{red}}
\def\bB{\begin{color}{blue}}
\def\bM{\begin{color}{magenta}}
\def\bC{\begin{color}{cyan}}
\def\bW{\begin{color}{white}}
\def\bBl{\begin{color}{black}}
\def\bG{\begin{color}{green}}
\def\bY{\begin{color}{yellow}}
\def\e{\end{color}\xspace}
\newcommand{\bit}{\begin{itemize}}
\newcommand{\eit}{\end{itemize}\par\noindent}
\newcommand{\ben}{\begin{enumerate}}
\newcommand{\een}{\end{enumerate}\par\noindent}
\newcommand{\beq}{\begin{equation}}
\newcommand{\eeq}{\end{equation}\par\noindent}
\newcommand{\beqa}{\begin{align*}}
\newcommand{\eeqa}{\end{align*}}
\newcommand{\beqn}{\begin{align}}
\newcommand{\eeqn}{\end{align}\par\noindent}

\def\jR{\begin{color}{black}}
\def\jB{\begin{color}{black}}
\def\jM{\begin{color}{magenta}}
\def\jC{\begin{color}{cyan}}
\def\jW{\begin{color}{white}}
\def\jBl{\begin{color}{black}}
\def\jG{\begin{color}{green}}
\def\jY{\begin{color}{yellow}}

\title{Indistinguishability in general probabilistic theories}

\author[1,2]{John H.~Selby}
\author[ ]{Victoria J.~Wright\!}
\author[3]{M\'at\'e Farkas}
\author[1,4]{Marcin Karczewski}
\author[1,2,5]{Ana Bel\'en Sainz}
\affil[1]{International Centre for Theory of Quantum Technologies, University of  Gda{\'n}sk, 80-309 Gda{\'n}sk, Poland}
\affil[2]{Theoretical Sciences Visiting Program, Okinawa Institute of Science and Technology Graduate University, Onna, 904-0495, Japan}
\affil[3]{Department of Mathematics, University of York, Heslington, York, YO10 5DD, United Kingdom}
\affil[4]{Institute of Spintronics and Quantum Information, Faculty of Physics, Adam Mickiewicz University, 61-614, Pozna\'n, Poland}
\affil[5]{Basic Research Community for Physics e.V., Germany}

\date{\today}

\begin{document}

\maketitle
\begin{abstract}
The existence of indistinguishable quantum particles provides an explanation for various physical phenomena we observe in nature. We lay out a path for the study of indistinguishable particles in general probabilistic theories (GPTs) via two frameworks: the traditional GPT framework and the diagrammatic framework of process theories. In the first approach we define different types of indistinguishable particle by the orbits of symmetric states under transformations. In the diagrammatic approach, we find a decomposition of the symmetrised state space using two key constructions from category theory: the biproduct completion and the Karoubi envelope. In both cases for pairs of indistinguishable particles in quantum theory we recover bosons and fermions. 

\end{abstract}

\tableofcontents

\section{Introduction}

Consider any pair of elementary particles of the same species, e.g., a pair of electrons or a pair of photons. 
Each of the two particles will have identical intrinsic properties such as mass, charge, and spin. If one wishes to distinguish two such identical particles, it follows that one must turn to their extrinsic properties, such as position, momentum, or relative spin direction. 

Quantum theory tells us, however, that pairs of identical particles 
reside in states with \emph{exchange symmetry}, whereby exchanging the labels of the particles results in a state that is indistinguishable from the original state. Explicitly, let $\cH$ be the Hilbert space, with orthonormal basis $\left\{\psi_j\right\}$, describing the state of each particle. Then, the pair can be in a state $\Psi=\sum_{j,k}c_{j,k}\psi_j\otimes \psi_k\in\cH\otimes\cH$ such that
\begin{equation}\label{eq:swap}
P\Psi=\sum_{j,k}c_{k,j}\psi_j\otimes \psi_k=e^{i\theta}\sum_{j,k}c_{j,k}\psi_j\otimes \psi_k=e^{i\theta}\Psi
\end{equation}
where $P$ is the swap operator defined by $P\psi\otimes\phi=\phi\otimes\psi$ for all $\psi,\phi\in\cH$. Since $\Psi$ and $P\Psi$ only differ by a global phase, the statistics given by measuring any hypothetical observable from quantum theory on either state $\Psi$ or $P\Psi$ would be identical, and thus, the two states cannot be distinguished. 

This invariance of the state\footnote{We say invariance here in the sense that the state is less redundantly described by a ray of projective Hilbert space, meaning the vectors $\Psi$ and $e^{i\theta}\Psi$ are two representatives of the same state.} under exchange of the labels will form the notion of indistinguishability we will study. Namely, a collection of identical particles are \emph{indistinguishable} if there is no observable difference between any relabelling of the particles. Note that under this notion of indistinguishability, models of physics considered to be classical also allow for pairs of systems to be indistinguishable. For example, consider a pair of identical point particles $A$ and $B$ with positions $x$ and $y$ and momenta $p$ and $q$ in $\R^3$, respectively. If we describe the state of the pair by the (unordered) set $\{(x,p),(y,q)\}$ then relabelling the particles results in the same state. 

There is, however, a further notion of indistinguishability not satisfied by this classical example. In this notion,  the two point particles can be distinguished by tracking their trajectories after some time at which the particles are first labelled. Theoretically, one can always measure their position with sufficient frequency to distinguish the particles, since they cannot occupy the same point in space and assuming they cannot travel faster than light. In quantum theory, however, there are no such trajectories. Indeed, the extrinsic properties which can be tracked in the classical case do not take fixed values that can be simultaneously observed without disturbing the state of the system at each point in time (see, e.g. Ref.~\cite{omar2005indistinguishable} for an introduction to this notion of indistinguishability). Within quantum theory one can demonstrate the trajectory indistinguishability follows from indistinguishability based on exchange symmetry~\cite{hugimb}. However, the point particle example demonstrates that this implication does not hold more broadly. In this work, the notion of indistinguishability based on exchange symmetry will remain our focus. 

Indistinguishable particles were described early in the development of quantum theory. For example, they provided Pauli's solution~\cite{Pauli1925} to Bohr's question~\cite{Falicov1977} of why all the electrons of an atom do not occupy the lowest energy level, which became the Pauli exclusion principle~\cite{Pauli1925a}, which was in turn generalised in the symmetrisation postulate~\cite{Messiah}. 

Although any phase in Eq.~\eqref{eq:swap} results in theoretically indistinguishable particles, empirically we have found that particles from the standard model exist only in symmetric or antisymmetric states, i.e., this phase only takes the values $\pm1$. Particles that are restricted to residing in symmetric states, such as photons and gluons, are known as bosons, while their antisymmetric counterparts, such as electrons and protons, are known as fermions. 

As put by Schr\"odinger~\cite{schrodinger}, the discovery of indistinguishable particles prompted physicists to
\begin{quote} \textit{dismiss the idea that a particle is an individual entity which retains its ‘sameness’ forever. Quite the contrary, we are now obliged to assert that the ultimate constituents of matter have no ‘sameness’ at all.} 
\end{quote}
There is now a great literature on the interpretation of the indistinguishability presented by quantum theory, discussing topics such as, whether indistinguishable particles can be considered as \emph{individuals} and the relation to Leibniz's principle of identity of indiscernibles  (see, e.g.,  Refs.~\cite{omar2005indistinguishable,french,Catren2023} and references therein). The information-theoretic potential of the indistinguishability of quantum particles is also now being established~\cite{app1,app2}.

This work will ask what we can learn about indistinguishability in our universe by considering how indistinguishable particles would present themselves in other hypothetical theories of physics, and whether there is anything special about quantum indistinguishability viewed in this broader landscape. We will study this question in the framework of \emph{general probabilistic theories} (GPTs)~\cite{Barrett,plavala2023general}, a class of theories derived from minimal assumptions about what procedures one can perform in a laboratory. 

GPTs have proven to be a versatile tool for studying both the fundamental physical and the information-theoretic characteristics of quantum and classical theories. The framework has been used, for example, to provide alternative and less abstract postulates for quantum theory~\cite{hardy2001quantum,Masanes}, to connect state space structure to degree of Bell non-locality~\cite{Janotta}, to show equivalence between generalised noncontextuality and other notions of classicality~\cite{Schmid}, and to demonstrate connections between non-classical features of quantum theory in a theory-independent way, including purification, reversible transformation, teleportation, and no-cloning~\cite{Chiribella}.

We now extend this line of study to include indistinguishable particles. This work aims to lay the ground work for studying the possibility of indistinguishable particles in GPTs. We propose approaches define how different types of indistinguishable particles emerge from a GPT, then examine these types in various GPTs. The goal of this analysis is establishing their similarities to and differences from the quantum case, to examine in what ways quantum indistinguishability is special. We will define indistinguishable particles using two different formalisms for the GPT framework. The first, more traditional approach describes GPTs by assigning to each system a state and effect space that are subsets of a real vector space (see, e.g.~\cite{Barrett}), while the second presents these concepts via a diagrammatic formulation of category theory, resulting in more visually-intuitive proofs (see, e.g.~\cite{Selby}). 

Earlier work on particle types in the study of GPTs \cite{dahlsten2013particle}, focused on how to measure the particle type via a ``swap experiment''. In this case particle types are labelled by elements of the phase group for a system. This approach, however, is limited in scope to the class of GPTs in which such a swap experiment is well defined. The broadest class known to admit of such an experiment is defined in Ref.~\cite{lee2016generalised}, however, following Ref.~\cite{barnum2019strongly} it seems likely that this is in fact a rather narrow class consisting only of finite dimensional Euclidean Jordan Algebras. In contrast, the definitions that we put forwards here are applicable to arbitrary GPTs. \blk


\section{Indistinguishability in the traditional GPT framework}\label{sec:trad}

General probabilistic theories (GPTs) describe (potentially hypothetical) physical systems from an operational perspective. The framework focuses on describing state preparation, measurement and the probabilities associated to measurement outcomes. This approach is firmly rooted in observable phenomena, abstracting away assumptions that cannot be directly verified. It consists of three core elements: states, effects (describing possible measurement outcomes) and transformations. We will briefly describe the basic components of this GPT framework, for motivation and more details, see e.g. Refs~\cite{hardy2001quantum,Barrett,plavala2023general}.

In this framework, \emph{states} of a system are equivalence classes of preparation procedures that lead to the same probabilities for all possible measurements. States form a convex subset, $\cS$, of a real vector space, $V$. Here convexity reflects possibility of probabilistically performing different preparation procedures. The extremal points of this set can be thought of as counterparts of the pure states from quantum theory. 

For (finite-dimensional) quantum theory the real vector space is the Hermitian operators on $\C^d$ and state space is the set of density operators, i.e. the non-negative operators $\rho\ge0$ with unit trace.

Similarly, \emph{effects} on a system represent equivalence classes of operationally indistinguishable measurement outcomes. Each effect is a linear functional that assigns a probability to each state, reflecting the likelihood of the associated outcome. The space of effects is bounded by the constraints on these probabilities, which must be non-negative and cannot exceed one. It follows that we can describe the set of effects, $\cE$, on a system by a subset of the vectors $\be\in V$ such that $0\le \langle \be,\bo\rangle\le 1$ for all states $\bo\in\cS$. 

Accordingly, in quantum theory effects are the Hermitian operators, $E$, satisfying $0\le\tr(E\rho)\le 1$ for all density operators $\rho$. Quantum theory satisfies the no-restriction hypothesis for effects, meaning that every vector satisfying the described constraints is included in the effect space. For simplicity, in this work we will assume our GPTs satisfy the no-restriction hypothesis.

A GPT also describes how systems compose. For example, consider a GPT $\mathfrak{G}$ that models a particle. Let the state (and effect) space $\mathcal{S}$ ($\mathcal{E}$) of the particle be embedded in a real vector space $V$. The GPT $\mathfrak{G}$ will then prescribe a state (and effect) space $\mathcal{S}^2$ ($\mathcal{E}^2$) for a pair of these particles embedded in $V\otimes V$, under the assumption of local tomography\footnote{Without local tomography the composed state and effect spaces would be embedded in $V\otimes V\oplus W$ for some other real vector space $W$. Allowing for these extra degrees of freedom would not affect our analysis so we omit them.}.

Finally, \emph{transformations} in GPTs describe how states evolve under controlled operations. Transformations are represented by invertible linear maps, $T:V\rightarrow V$, that preserve the state space, i.e. $T\bo\in\cS$ for all $\bo\in\cS$. Again, for simplicity we will assume no-restriction on the set of transformations. We also impose the condition that if we compose a system with any other system in the GPT and act with $T$ tensored with the identity, the resulting vector should be a valid state of the composed system. 

In quantum theory, these conditions combine to identify transformations as invertible CPTP maps. Since we are only considering transformations of a given state space the dimension of the Hilbert space is fixed and these maps correspond to unitary transformations \cite{Nayak2007},
\begin{equation}\label{eq:unitary}
T(\rho)=U\rho U^\dagger
\end{equation}
for some unitary operator $U$ on $\cH$.

In general, the transformations preserving a state space $\mathcal{S}$ and satisfying the composition rule for the GPT $\mathfrak{G}$ form a group $\mathcal{T^\mathfrak{G}_\mathcal{S}}$ (we will drop the index $\mathfrak{G}$ to simplify notation).

This framework is derived from minimal assumptions, and thus GPTs encompass a broad range of theories from those in classical physics to hypothetical theories reaching beyond quantum mechanics. They provide a testbed for comparing fundamental properties of physics. In this section we will use this framework to investigate the notion of indistinguishable particles.


\subsection{Symmetries and groups in GPTs}

We will now motivate and define the mathematical tools we will use to identify the indistinguishable particles that emerge within a theory. Consider two particles each with state space $\cS\subset V$, and assume that these particles are indistinguishable. The indistinguishablility assumption means that if we relabel (swap) the systems, the state of the pair should be preserved. Thus, we consider the subset $\mathcal{S}^2_s$ of states in $\mathcal{S}^2$ that are also in the symmetric subspace of $V\otimes V$. That is, if $\left\{\be_1,\ldots,\be_n\right\}$ is a basis for $V$  consider states $\bo$ such that
\begin{equation}\label{eq:bisym}
\bo=\sum_{ij}v_{ij}\be_i\otimes\be_j=\sum_{ij}v_{ij}\be_j\otimes\be_i.
\end{equation}
This set of symmetric states can be also expressed as  
\begin{equation}\label{eq:symmetric_states}
\mathcal{S}^2_s = \{ \bo \in \mathcal{S}^2 ~|~ P \bo = \bo \}\,,
\end{equation}
 where $P$ is the swap operator 
\begin{equation}\label{eq:swap}
P( \be_i\otimes\be_j ) = \be_j\otimes\be_i\,.
\end{equation}
We will generally refer to such states as symmetric, and pure quantum states satisfying Eq.~\eqref{eq:swap} as wavefunction symmetric to disambiguate between the two notions.

Now, if we act on a pair of indistinguishable particles with some transformation $T \in \mathcal{T}_\mathcal{S}$, then the state of the system should remain symmetric  after $T$ is applied (otherwise the particles become distinguishable). This means that the set of operations that can be applied to a system of indistinguishable particles should be restricted to the subset of $\mathcal{T}_\mathcal{S}$ that preserves the symmetric subspace of $V$. 
That is, under the assumption of indistinguishability the transformations will be given by the subgroup $\mathcal{T}^s_\mathcal{S}\leq\mathcal{T}_\mathcal{S}$ given by the transformations that commute with the swap operator on the symmetric subspace, 
\begin{equation}\label{eq:symmetry_preserving}
\mathcal{T}^s_\mathcal{S} = \{ T \in \mathcal{T}_\mathcal{S} ~|~ [T, P]_s = 0 \}.
\end{equation}

\subsection{Particle types}\label{types}

Under these symmetric transformations, $\cT_\cS^s$, it will generally not be possible to transform every symmetric pure state into any other symmetric pure state. This means under our notion of indistinguishability a system will be trapped in some part of the space of symmetric states. These disjoint areas of the state space form the intuition for our notion of particle types. The type of particle will be defined by which subset of the symmetric states it can occupy.


When defining the state spaces of the different particle types, we want to avoid states deriving from mixtures of states of different particle types. \blk 
For instance, in the case of quantum theory, mixing bosonic and fermionic states can give the maximally mixed state whose density operator satisfies Eq.~\eqref{eq:bisym} but is not a valid state of a boson or fermion. 

In order to avoid including such mixtures in the state spaces of the different particle types, initially we limit our consideration finding the pure states of each particle type. We have two options: 
\begin{enumerate}
\item[I.] the extremal symmetric states (extremal points of the set of symmetric states), or
\item[II.] the symmetric extremal states (the subset of extremal points of the full state space that are symmetric). 
\end{enumerate}
In quantum theory these two options coincide.

\begin{lemma}
The pure states $\psi\in\cH\otimes\cH$ satisfying $P\psi=e^{i\theta}\psi$ (the symmetric extremal states) are the extremal points of the set of density operators $\rho$ on $\cH\otimes\cH$ that satisfy $P\rho P= \rho$ (extremal symmetric states).
\end{lemma} 
\begin{proof}
The Hilbert space $\cH\otimes\cH$ can be decomposed into the direct sum of its (wavefunction) symmetric and anti-symmetric subspaces, $\cH\otimes\cH=\cH_s\oplus\cH_a$. The (operationally) symmetric extremal states of the density operators on $\cH\otimes\cH$ are the pure states of these two subspaces, i.e. the states $\psi\in\cH\otimes\cH$ satisfying either $P\psi=\pm\psi$.

Given a density operator $\rho$ on $\cH\otimes\cH$ that satisfies $P\rho P= \rho$ we will show it can be written as a convex combination of the pure states described above. Let $S$ and $A$ be the orthogonal projections onto the symmetric and antisymmetric subspaces $\cH_s$ and $\cH_a$, respectively. Note that we have $S=(\I+P)/2$ and $A=(\I-P)/2$. It follows that
\begin{equation}
\begin{aligned}\label{eq:SplusA}
\rho &= (S+A)\rho(S+A)\\
&=S\rho S + S\rho A + A\rho S + A\rho A\\
&=S\rho S + A\rho A,
\end{aligned}
\end{equation} 
where the final line follows since 
\begin{equation}
S\rho A = \frac14(\rho + P\rho - \rho P - P\rho P) = \frac14(P\rho - \rho P)
\end{equation}
by the assumption that $P\rho P = \rho$, and similarly
\begin{equation}
A\rho S  = \frac14(-P\rho + \rho P). 
\end{equation}
By Eq.~\eqref{eq:SplusA}, $\rho$ can be written as the convex combination $$\rho= \tr(\rho S)\rho_s + \tr(\rho A)\rho_a$$ of the density operators $\rho_s=S\rho S/\tr(\rho S)$ and $\rho_a=A\rho A/\tr(\rho A)$. In turn, the operator $\rho_s$ ($\rho_a$) is a density operator on $\cH_s$ ($\cH_a$) and can therefore be written as a convex combination of the pure states in this Hilbert space, i.e. the symmetric (antisymmetric pure states) on $\cH\otimes\cH$.
\end{proof}

However, in other GPTs these two sets of states may not be equivalent. For example, consider the classical two-bit GPT. The full state space has four orthogonal extremal states $[00]$, $[11]$, $[01]$ and $[10]$, the first two of which are symmetric. However, the subset of symmetric states also contains the state $([01]+[10])/2$ which is not in the convex hull of $[00]$ and $[11]$. So the symmetric extremal states $[00]$ and $[11]$ are not all the extremal points of the set of symmetric states. See Sec.~\ref{sec:class}. We leave as a question for future work, in which GPTs apart from quantum theory are the symmetric extremal and extremal symmetric states the same?

In the following $\mathcal{S}^2_{\text{ex},s}$ can denote either of these two choices for the symmetric pure states. 
First, we consider defining particles types in \emph{transitive} GPTs (such as quantum theory), wherein any pure state can be transformed into any other. In other words the set of orbits 
\begin{equation}
\mathcal{O}_{\mathcal{T}_\mathcal{S}} := \mathcal{S}^2_\text{ex} / \mathcal{T}_\mathcal{S}\,.
\end{equation}
has only one element. Here, the orbit of an extremal state $\bo_\text{ex}$ under the group $\mathcal{T}_\mathcal{S}$ is given by the set
\begin{equation}
[ \bo_\text{ex} ] = \left\{\bo ~|~ \bo = T \bo_\text{ex}\text{ for some }T\in\mathcal{T}_\mathcal{S}\right\}.
\end{equation}

In these transitive GPTs, to define different types of indistinguishable particles we look at which extremal symmetric states $\bo_\text{ex}$ can be transformed into one another by the transformations $\mathcal{T}^s_\mathcal{S}$. The orbit of an extremal state $\bo_\text{ex}$ under the group $\mathcal{T}^s_\mathcal{S}$ is given by the set
\begin{equation}
[ \bo_\text{ex} ] = \left\{\bo ~|~ \bo = T \bo_\text{ex}\text{ for some }T\in\mathcal{T}^s_\mathcal{S}\right\}.
\end{equation}

We can then define the different types of particles to be given by the orbits 
\begin{align}
\mathcal{O}_{\mathcal{T}^s_\mathcal{S}} := \mathcal{S}^2_{\text{ex},s} / \mathcal{T}^s_\mathcal{S}\,.
\end{align} 

That is,
\begin{quote}
we interpret \textit{$\mathcal{O}_{\mathcal{T}^s_\mathcal{S}}$ as  the set of different types of indistinguishable particles} predicted by the GPT $\mathfrak{G}$, since only the states in each orbit can be transformed into each other under the indistinguishablility condition.
\end{quote}

We denote the different types of particles by $\pi$, and therefore $\mathcal{S}^2_{\text{ex},s} / \mathcal{T}^s_\mathcal{S} = \{ [ \bo_\pi ] \}_\pi$. We recover the full state space of a given type of particle by taking the convex hull of the extremal points,
\begin{equation}\label{eq:particle_type}
\mathcal{S}^2_\pi = \Conv [\bo_\pi] .
\end{equation}

In \textit{non-transitive GPTs}, such as Boxworld \cite{Barrett}, the set $\mathcal{O}_{\mathcal{T}_\mathcal{S}}$ may have already more than one orbit, and, hence, so will $\mathcal{O}_{\mathcal{T}^s_\mathcal{S}}$. The orbits in $\mathcal{O}_{\mathcal{T}^s_\mathcal{S}}$ can then not straightforwardly be understood as different particles, since the structure may be inherent to the GPT itself rather than emerging from symmetry considerations. In these cases we examine how each of the existing orbits of $\mathcal{O}_{\mathcal{T}_\mathcal{S}}$ split into symmetric orbits\footnote{We see that orbits $\mathcal{O}_{\mathcal{T}^s_\mathcal{S}}$ are splittings of the original orbits $\mathcal{O}_{\mathcal{T}_\mathcal{S}}$ rather that some sort of mixing, by thinking of these orbits as equivalence classes. Then,  \blk when constructing the equivalence classes in $\mathcal{O}_{\mathcal{T}^s_\mathcal{S}}$ one has access to fewer transformations than when constructing the equivalence classes of $\mathcal{O}_{\mathcal{T}_\mathcal{S}}$, and hence naturally some of the orbits of $\mathcal{O}_{\mathcal{T}^s_\mathcal{S}}$ will become equivalent in $\mathcal{O}_{\mathcal{T}_\mathcal{S}}$.} to create the set $\mathcal{O}_{\mathcal{T}^s_\mathcal{S}}$. We therefore can conceivably end up with a different set of particles for each of the original orbits in $\mathcal{O}_{\mathcal{T}_\mathcal{S}}$.\blk

 To formalise this, we define a (potentially non-injective) function $F:\mathcal{O}_{\mathcal{T}^s_\mathcal{S}} \to \mathcal{O}_{\mathcal{T}_\mathcal{S}}$ which maps a symmetric orbit into the original orbit that contains it, that is, such that $x\subseteq F(x)$ for all symmetric orbits $x\in \mathcal{O}_{\mathcal{T}^s_\mathcal{S}}$. 
 In this case, 
\begin{quote}
we interpret the preimage $F^{-1}(y)$ as the set of different types of indistinguishable particles predicted by the GPT $\mathfrak{G}$ for a given orbit $y$. In the transitive case we have just a single orbit $y$ hence we recover the previous definition.
\end{quote}
\blk

\subsection{Examples}\label{se:examples}

\subsubsection{Quantum theory}

In quantum theory, we recover the boson and fermion particle types. Extremal states are rank-1 projections onto the subspaces of the Hilbert space $\mathcal{H}=\C^d\otimes\C^d$. Due to the Schur--Weyl duality, $\mathcal{H}$ can be decomposed as $\mathcal{H}=\mathcal{H}_s\oplus\mathcal{H}_a$ into the direct sum of its wavefunction symmetric and antisymmetric subspaces (corresponding to the trivial and the sign representation of the symmetric group of order 2). The unitary maps corresponding to the transformations in the subgroup $\mathcal{T}^s_{Q_n}$ (as in Eq. \eqref{eq:unitary}) will all be of the form $U_s\oplus U_a$, where $U_s$ is a unitary on $\mathcal{H}_s$ and $U_a$ a unitary on $\mathcal{H}_a$ (this follows from $U$ commuting with the swap operator $P$). Hence there will be two orbits of $\mathcal{S}^2_{\text{ex},s}$ consisting of the wavefunction symmetric and antisymmetric states, i.e., two types of particles, bosons and fermions:
\begin{equation}\label{eq:boson-fermion}
\mathcal{S}^2_{\text{ex},s} / \mathcal{T}^s_{Q_n} = \Big\{ \big[ \ketbraq{ \psi_s } \big], \big[ \ketbraq{ \psi_a } \big] \Big\},
\end{equation}
where $\ket{ \psi_{s(a)}}$ is any wavefunction symmetric (antisymmetric) unit vector on $\C^d\otimes\C^d$.

For two qubits, the three symmetric Bell states form a basis for the symmetric subspace of $\C^2\otimes\C^2$ and the singlet is the only antisymmetric pure state. So if you write a unitary that preserves the symmetric subspace of the Hermtian operators on $\C^2\otimes\C^2$ in the Bell basis it will have the form $U\oplus\mathrm{e}^{\mathrm{i}\varphi}$ for a $3\times3$ unitary $U$.

\subsubsection{Classical GPT}\label{sec:class}

Consider the GPT in which the state spaces for each dimension $d\in\mathbb{N}$ are given by simplices with orthogonal extremal points $[0],[1],\ldots,[d-1]$. Transformations $T$ are simply relabellings of the $d$ pure states. A pair of such systems of the same dimension has a simplicial state space with orthogonal extremal point $[j]\otimes[k]$ for $0\le j,k \le d-1$. The symmetric states are those in the convex hull of the states $1/2([j]\otimes[k] +[k]\otimes[j])$ for all $0\le j,k \le d-1$. The transformations that preserve this space are those of the form $T\otimes T$.  

For option I in Sec.~\ref{types} we take the pure states of the particle types to be all the states $1/2([j]\otimes[k] +[k]\otimes[j])$ for $0\le j,k\le d-1$. We get two orbits under the transformations $T\otimes T$, one comprising the extremal states with $j=k$, i.e. the states $[j]\otimes[j]$ and the other comprising the states $1/2([j]\otimes[k] +[k]\otimes[j])$ for $j\neq k$. Hence, this approach gives two types if indistinguishable particle for each dimension.

On the other hand, for option II in Sec.~\ref{types} we take the pure states of the particle types to be only the states $[j]\otimes[j]$. In this case we only get one orbit and hence particle type, since $[j]\otimes[j]$ can be taken to $[k]\otimes[k]$ $0\le j,k\le d-1$ by a transformation $T\otimes T$.

\subsubsection{Boxworld}

In the GPT called Boxworld \cite{Barrett}, \blk the reversible transformations were identified in Ref.~\cite{Gross2010} to be only those given by local relabellings of the measurements and outcomes and relabellings of the systems. Since such transformations are non-entangling, it follows that Boxworld is non-transitive. In the simplest case of two particles, each associated with a square state space, one finds two orbits: one consisting of only product states, and one of only entangled. If one then searches for the orbits under the subgroup of transformations that commute with the swap operator, one does not find new orbits. The same two orbits as before (product states and entangled states) are recovered (each fewer states, however, since we only take the symmetric states). 
In this approach, Boxworld does not admit of different indistinguishable particle types, beyond the particles types already present without the assumption of indistinguihability.

\subsubsection{Spekkens' toy model} 
In the Spekkens' toy model~\cite{Spekkens} a system resides in one of four \emph{ontic} states denoted `1', `2', `3' and `4'. Characteristically, no measurement can fully determine the ontic state of the system. The model satisfies the knowledge balance principle: the maximal attainable knowledge cannot exceed the amount of lacking information. This principle gives rise to a state space labelled by the pairs of possible ontic states underlying them linked with the $`\vee$' symbol. The states corresponding to the maximal knowledge of the elementary system (\emph{pure} states) are: $1\vee2$, $1\vee3$, $1\vee4$, $2\vee3$, $2\vee4$ and $3\vee4$. This theory can be convexified in a GPT as considered in Ref.~\cite{Janottaa}, resulting in a state and effect space both given by an octahedron. The allowed transformations are defined as permutations of the ontic states which respect the knowledge balance principle. Notably, this GPT system does not satisfy the no-restriction hypothesis, but nonetheless we can apply our methodology.

The ontic states of two elementary systems are defined by the ontic states of the subsystems. This leads to 16 ontic states denoted $1\cdot1,\,1\cdot2,  \ldots, 4\cdot4$.  All the pure epistemic states can be obtained by local permutations from the separable state $(1\cdot1)\vee(1\cdot2)\vee(2\cdot1)\vee(2\cdot2)$ and entangled state $(1\cdot1)\vee(2\cdot2)\vee(3\cdot3)\vee(4\cdot4)$, which can be expressed in the following diagrams wherein a square is shaded if it corresponds to a possible ontic state
\begin{equation}
 \begin{matrix}   \includegraphics[scale=0.5]{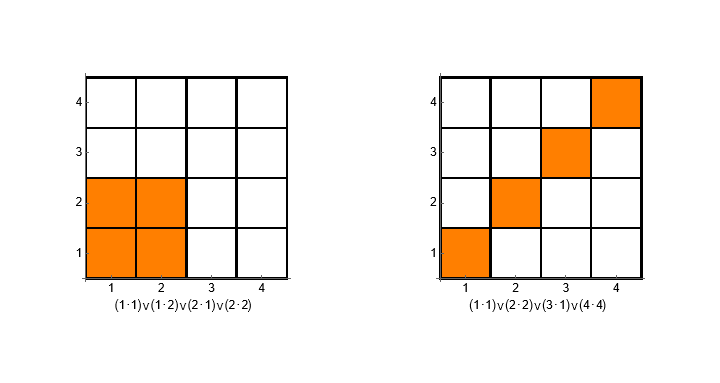}.
\end{matrix}
\end{equation}

Of these states the following 16 are invariant under exchange of the systems
\begin{equation}
\label{spekkensOrbits}
\begin{matrix}
    \includegraphics[scale=0.7]{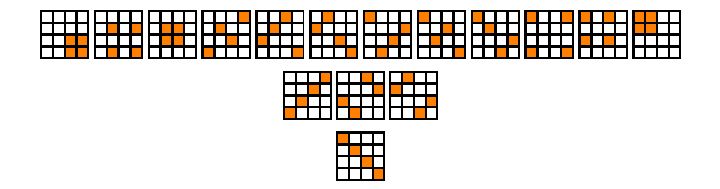}
\end{matrix}
\end{equation}
and they form our symmetric extremal states.

The set of symmetric transformations $\mathcal{T}^s_\mathcal{S}$ consist of applying the same local permutation of ontic states to both subsystems and an entangling gate defined graphically as
\begin{equation}
\label{cnot}
    \begin{matrix}\includegraphics[scale=0.3]{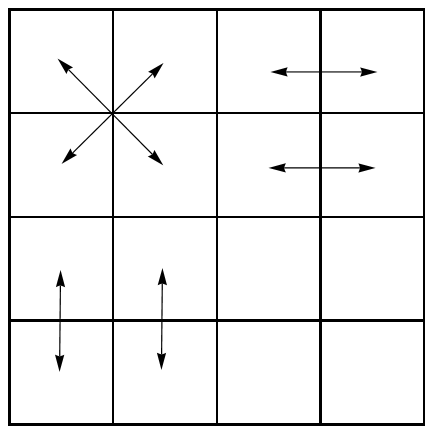}
    \end{matrix}.
\end{equation}
Note that the transformation (\ref{cnot}) can be obtained from the CNOT transformation defined in~\cite{Spekkens} by a swap of columns.
These transformations split our symmetric extremal states into three orbits, characterized by the number of symmetric ontic states forming the epistemic states of each class: 0, 2 or 4 (see the rows of Eq.~(\ref{spekkensOrbits})). These three orbits define the types of indistinguishable particles, under option II from Sec.~\ref{types}

\section{Indistinguishability in the diagrammatic approach}

Now, we will consider the question of indistinguishable particles in GPTs using the formalism of process theories \cite{coecke2018picturing}.
In the following subsections we will get into the details of what process theories are, and how indistinguishable particles emerge in the formalism. Here let us summarise what will unravel.

 A key result from Ref.~\cite{Coecke} that will help us frame the property of indistinguishability, is the relation of two categorical constructions: the \emph{biproduct completion} and the \emph{Karoubi envelope}. One finds that symmetrisation at the operational level (as expressed within the Karoubi envelope) leads to the decomposition of the symmetrised multipartite state space into a direct sum (a.k.a. biproduct) of state spaces. 
If it is the case (such as in quantum theory) that prior to symmetrisation there is no nontrivial direct sum decomposition, then we will think of each of these direct summand state spaces as the multipartite particle types. The case where there is already a nontrivial direct sum decomposition (such as in classical theory) then the situation is somewhat more complicated, but essentially particle types correspond to any new decompositions which emerge through the symmetrisation process. \blk
 
 In the particular case of quantum theory, for pairs of particles these sectors can be shown to correspond to bosons and fermions. For quantum systems of more particles, other sectors emerge, which we interpret as paraparticles. How many sectors and what sorts of particles we get in more general settings, i.e., beyond quantum and classical theory, is an open question.

\subsection{Process theories}

 A process theory is defined by a collection of systems, processes that act on these systems, and rules on how these systems  may compose. Different ways of specifying these ingredients give rise to a different process theory. For example, a way to specify quantum theory as a process theory is by taking a system to be a `Hilbert Space', the processes to be `completely positive and trace preserving maps', and composition rule to be the tensor product. A process theory comes equipped with a diagrammatic calculus, where one can represent these systems by `wires', and the processes acting on them by `boxes', like in the example of Eq.~\eqref{eq:diagex}. 

Process theories  can be thought of  as being equivalent to symmetric monoidal categories (SMCs). Essentially, the elements of an SMC correspond to the wires and boxes in the diagrams below, and the operations of the SMC correspond to how the wires and boxes are turned into diagrams. The axioms  for the SMC ensure that the diagrams make sense: that is, two different ways of writing a diagram in terms of the basic SMC operations are necessarily equal by the axioms of the SMC.

In diagrams, such as:
\beq\label{eq:diagex}
\InputIfFileExists{Diagrams/diagram.tikz}{}{\input{./figures/Diagrams/diagram.tikz}},
\eeq
processes (i.e., boxes) with no inputs (such as $s$) are states, and those with no outputs (such as $e$) are effects. Note also that the swapping of the systems (i.e., wires) $A$ and $C$  in this diagram is an allowed operations, which makes this a \emph{symmetric} monoidal category rather than just a monoidal category. This swap operation satisfies various conditions but most relevant for us is that it is self-inverse:
\beq
\InputIfFileExists{Diagrams/swap.tikz}{}{\input{./figures/Diagrams/swap.tikz}}\quad=\quad%
\InputIfFileExists{Diagrams/swap2.tikz}{}{\input{./figures/Diagrams/swap2.tikz}}\,.
\eeq
Intuitively this means that wires do not get into knots and instead can simply pass through each other. This gives us with a natural way to represent the permutation group on a composite system. For example, in the above diagram we have the permutation group of two elements given by the identity and the swap, such that composing the swap with itself is the identity.

 Another important ingredient in the formalisation of the process theories we are interested in is sums of processes\footnote{ Not every process theory comes with sums, but all of the ones of interest for us such as quantum theory, classical theory, and general probabilistic theories, do.\blk}. The key feature of summing is that it distributes over diagrams, that is:
\beq
\InputIfFileExists{Diagrams/diagSum1.tikz}{}{\input{./figures/Diagrams/diagSum1.tikz}} \quad = \quad %
\InputIfFileExists{Diagrams/diagSum2.tikz}{}{\input{./figures/Diagrams/diagSum2.tikz}}. 
\eeq
Closed diagrams are those with neither inputs nor outputs, and are usually known as scalars.  For the theories of interest to us, \blk these are identified with probabilities.

We can now define mixtures of processes:
\beq
\InputIfFileExists{Diagrams/mixture.tikz}{}{\input{./figures/Diagrams/mixture.tikz}}.
\eeq
We will often drop the `box' around scalars, and simply write in the diagram the probability $p_i$.

Having defined mixtures we can now define the following operational symmetrisation process  for two systems:
 \beq
\InputIfFileExists{Diagrams/symmeterise1.tikz}{}{\input{./figures/Diagrams/symmeterise1.tikz}}\quad := \quad \frac{1}{2} \ \ %
\InputIfFileExists{Diagrams/identAA.tikz}{}{\input{./figures/Diagrams/identAA.tikz}}\ \  +\  \frac{1}{2} \ \ %
\InputIfFileExists{Diagrams/swapAA.tikz}{}{\input{./figures/Diagrams/swapAA.tikz}} .
 \eeq
This diagram represents a process wherein with equal probability we either leave the particles alone or we swap them over. We call this operational symmetrisation rather than just symmetrisation since, in quantum theory, this process is symmetrisation at the level of the density matrix rather than at the level of the wavefunction. Operational symmetrisation also generalises to arbitrary numbers of particles, by taking equal mixtures of all the possible exchanges of the particles.  Each possible way to exchange particles in a multipartite system corresponds to a permutation $\Pi$ of the systems, can always be just drawn diagrammatically as a wiring of the input wires to the output wires. 

There are two key features of these (multipartite) symmetrisation processes. The first is that they are invariant under permutations, $\Pi$:
\beq\label{eq:permInvariance}
\InputIfFileExists{Diagrams/multiSymPerm.tikz}{}{\input{./figures/Diagrams/multiSymPerm.tikz}}\quad = \quad %
\InputIfFileExists{Diagrams/multiSym.tikz}{}{\input{./figures/Diagrams/multiSym.tikz}}\quad = \quad%
\InputIfFileExists{Diagrams/multiSymPerm2.tikz}{}{\input{./figures/Diagrams/multiSymPerm2.tikz}}.
\eeq
The second is that they are idempotent, i.e.:
\beq
\InputIfFileExists{Diagrams/multiSymIdem.tikz}{}{\input{./figures/Diagrams/multiSymIdem.tikz}}\quad = \quad %
\InputIfFileExists{Diagrams/multiSym.tikz}{}{\input{./figures/Diagrams/multiSym.tikz}}.
\eeq

\subsection{Splitting idempotents}

 A concept that will be crucial for us when defining indistinguishable particles is  the idea of `splitting an idempotent'. 

\begin{definition}
An idempotent $P\!:\!A\!\to\!A$ splits through a system $B$ iff there exist processes $\iota:A\to B$ and $\kappa:B\to A$ such that:
 \beq
\InputIfFileExists{Diagrams/idemP.tikz}{}{\input{./figures/Diagrams/idemP.tikz}}\quad = \quad %
\InputIfFileExists{Diagrams/idemPSplit.tikz}{}{\input{./figures/Diagrams/idemPSplit.tikz}}\qquad \text{and} \qquad %
\InputIfFileExists{Diagrams/idemPSplit2.tikz}{}{\input{./figures/Diagrams/idemPSplit2.tikz}}\quad = \quad %
\begin{tikzpicture}
	\begin{pgfonlayer}{nodelayer}
		\node [style=none] (33) at (0, -1.5) {};
		\node [style=none] (40) at (0, 1.5) {};
		\node [style=right label] (41) at (0, -1) {$B$};
	\end{pgfonlayer}
	\begin{pgfonlayer}{edgelayer}
		\draw [qWire, in=90, out=-90] (40.center) to (33.center);
	\end{pgfonlayer}
\end{tikzpicture}
}.
 \eeq
\end{definition}
 
In the case of the symmetrisation process, the process it splits through can be thought of as the operationally symmetric subspace. For a pair of quantum systems, \blk the operationally symmetric subspace defined in this way for quantum theory will recover the direct sum of the wavefunction-symmetric and the wavefunction-antisymmetric spaces, i.e., the fermionic and  bosonic subspaces. Going beyond pairs of quantum systems we also find other direct summands which correspond to paraparticle subspaces. \blk

Now, in general not all idempotents in a given theory will split. In fact, this is the case in the standard presentation of the quantum GPT: the idempotent splits through the direct sum of fermionic and bosonic subspaces, but this direct sum is not typically considered to be itself a system  within quantum theory (i.e., it cannot simply be labelled by a Hilbert space). However, there is a well known categorical construction known as the Karoubi envelope (or idempotent completion, or Cauchy completion) which  can be applied to any process theory to define a new process theory in which all idempotents do split. This construction, crucial to our work, is presented in the next subsection. 

\subsection{Karoubi Envelope}

Given some process theory $\mathcal{P}$ we can define a new process theory $K[\mathcal{P}]$ in which all idempotents split. 
The systems in this new process theory  $K[\mathcal{P}]$ are labelled by the idempotents $P\!:\!A\!\to\!A$  where $A$ is an arbitrary system in $\mathcal{P}$ and $P\!:\!A\!\to\!A$ is an arbitrary idempotent process in $\mathcal{P}$.  The composition of two systems $P\!:\!A\!\to\!A$ and $Q\!:\!B\!\to\!B$  in $K[\mathcal{P}]$ is defined as
\begin{align}
(P\!:\!A\!\to\!A)\otimes (Q\!:\!B\!\to\!B) := P\otimes Q\!:\!A\otimes B\!\to\!A\otimes B
\end{align}
which is legitimate because (amongst other things) the process $P\otimes Q\!:\!A\otimes B\!\to\!A\otimes B$  in $\mathcal{P}$ is also idempotent in $\mathcal{P}$.\blk

Now, the processes in  $K[\mathcal{P}]$ with input system $P\!:\!A\!\to\!A$ and output system $Q\!:\!B\!\to\!B$  are labelled by the processes $T$ in the original process theory with input $A$ and output $B$  which are invariant under  composition with the idempotents, i.e., 
\beq\label{eq:KaroubiProcess}
\InputIfFileExists{Diagrams/karoubiCon1.tikz}{}{\input{./figures/Diagrams/karoubiCon1.tikz}}\quad = \quad %
\InputIfFileExists{Diagrams/karoubiCon2.tikz}{}{\input{./figures/Diagrams/karoubiCon2.tikz}},
\eeq
where this diagram should be read within the process theory $\mathcal{P}$. That is, if we have a process $T$ in $\mathcal{P}$ that does not satisfy Eq.~\eqref{eq:KaroubiProcess} for some choice of idempotents $P$ and $Q$ in $\mathcal{P}$, then $T$ cannot label a process from system $P\!:\!A\!\to\!A$ to system $Q\!:\!B\!\to\!B$ in  $K[\mathcal{P}]$. That same $T$, however, might label a process in $K[\mathcal{P}]$ but associated to different input/output systems.

Note that idempotence of the processes $Q$ and $P$  in  the process theory $\mathcal{P}$ means that we can define a mapping from the `processes from system $A$ to system $B$ in $\mathcal{P}$' into the `processes from system $P\!:\!A\!\to\!A$ to system $Q\!:\!B\!\to\!B$ in $K[\mathcal{P}]$' via:
\beq\label{eq:KaroubiMapExt}
\InputIfFileExists{Diagrams/karoubiMap1.tikz}{}{\input{./figures/Diagrams/karoubiMap1.tikz}}\quad \mapsto \quad %
\InputIfFileExists{Diagrams/karoubiMap2.tikz}{}{\input{./figures/Diagrams/karoubiMap2.tikz}}.
\eeq
That is, the process in $\mathcal{P}$ labelled by $f$ maps into the process in $K[\mathcal{P}]$ from type $P\!:\!A\!\to\!A$ to type $Q\!:\!B\!\to\!B$ labelled by $Q\circ f\circ P$.  Indeed, one can see that the RHS of Eq.~\eqref{eq:KaroubiMapExt} satisfies Eq.~\eqref{eq:KaroubiProcess}, and so  it labels a process  in $K[\mathcal{P}]$. 

 An important  example of a system in $K[\mathcal{P}]$ is one labelled by an identity transformations in $\mathcal{P}$, since the identity $\mathds{1}_A\!:\!A\!\to\!A$ is an idempotent operation   in $\mathcal{P}$ and hence labels a system in $K[\mathcal{P}]$. Condition \eqref{eq:KaroubiProcess} is then trivially satisfied, so the labels for  the processes from system $\mathds{1}_A\!:\!A\!\to\!A$ to system $\mathds{1}_B\!:\!B\!\to\!B$ in $K[\mathcal{P}]$ include all of the processes from $A$ to $B$ in $\mathcal{P}$. We therefore see that $\mathcal{P}$ is a full subtheory of $K[\mathcal{P}]$. It is therefore convenient to label the systems $\mathds{1}_A\!:\!A\!\to\!A$ simply by $A$ and to view $K[\mathcal{P}]$ as an extension of $\mathcal{P}$. Indeed, we will take this view below. 

Using this shorthand notation we can now see in what sense the Karoubi envelope allows all idempotents to split.
Notice, first, that any idempotent $P\!:\!A\!\to\!A$ in $\mathcal{P}$ corresponds to a process from $\mathds{1}_A:A\to A$ to $\mathds{1}_A:A\to A$ in $K[\mathcal{P}]$ labelled by $P$. Relabelling the systems $\mathds{1}_A:A\to A$ simply by $A$, then any idempotent $P$ in $K[\mathcal{P}]$ splits through the system $P\!:\!A\!\to\!A$ as follows: 
\beq\label{eq:thespl}
\InputIfFileExists{Diagrams/idemP.tikz}{}{\input{./figures/Diagrams/idemP.tikz}}\quad = \quad %
\InputIfFileExists{Diagrams/idemPSplitKaroubi.tikz}{}{\input{./figures/Diagrams/idemPSplitKaroubi.tikz}}.
\eeq
In fact, $\iota_P$ and $\kappa_P$ are actually the idempotent itself, but now interpreted as maps from $\mathds{1}_A:A \to A$ to $P\!:\!A\!\to\!A$ and from $P\!:\!A\!\to\!A$ to $\mathds{1}_A:A\to A$ respectively. 

 With these tools, then, one can explicitly specify the process that $f$ is mapped into as per Eq.~\eqref{eq:KaroubiMapExt}: 
\beq\label{eq:KaroubiMapInt}
\InputIfFileExists{Diagrams/karoubiMap1.tikz}{}{\input{./figures/Diagrams/karoubiMap1.tikz}}\quad \mapsto \quad %
\InputIfFileExists{Diagrams/karoubiMap3.tikz}{}{\input{./figures/Diagrams/karoubiMap3.tikz}} \,,
\eeq
where notice that the diagram depicted in the right-hand side of Eq.~\eqref{eq:KaroubiMapInt} is a process in $K[\mathcal{P}]$.

\subsection{Splitting symmetrisations}\label{sec:SplitSym}

For this work we are not interested in arbitrary idempotents in $\mathcal{P}$, but specifically in its symmetrisation idempotents. We therefore will not work with $K[\mathcal{P}]$ in full, but, instead with  full subtheory of $K[\mathcal{P}]$  which  we denote by \Ks and consists of: 
\begin{compactitem}
\item Atomic systems: $\mathds{1}_A\!:\!A\!\to\!A$ for any $A\in |\mathcal{P}|$ (which as mentioned above we will simply label by $A$) and $\mathsf{Sym}:\underbrace{A\otimes \cdots \otimes A}_n \to \underbrace{A\otimes \cdots \otimes A}_n$ for any $A\in |\mathcal{P}|$ and $n\in \mathds{N}$ which we will simply label by $\mathsf{Sym}_A^n$.
\item General systems are then any composite of the atomic systems.
\item Processes: the subset of those of $K[\mathcal{P}]$ that these as input and output systems.
\item Composition rule: same as $K[\mathcal{P}]$.
\end{compactitem}
Notice that in \Ks we may have systems composed of $n$ distinguishable particles of the same type $A$, or one system $\mathsf{Sym}_A^n$ composed of $n$ indistinguishable particles of the same individual type $A$. 
I an way, one can either see \Ks as a full subtheory of $K[\mathcal{P}]$, or as an extension of $\mathcal{P}$ to include these new system types that denote collections of indistinguishable particles. 

In \Ks then one has the symmetrisation idempotents which act on collections of systems, and which split through this new system type $\mathsf{Sym}_A^n$. In the diagrams hereon we will not label the processes that correspond to these idempotents, for simpliticy in the presentation. Other processes will be labelled. 

Within \Ks one can write equations like the following: 
\blk
\beq
\InputIfFileExists{Diagrams/symKaroubi.tikz}{}{\input{./figures/Diagrams/symKaroubi.tikz}}\quad =\quad %
\InputIfFileExists{Diagrams/symKaroubi1.tikz}{}{\input{./figures/Diagrams/symKaroubi1.tikz}}.
\eeq
There, we see how the idempotent of $n$ systems (left hand side) splits through system $\mathsf{Sym}_A^n$. Moreover, in the right hand side one can further interpret the embeddings $\iota$ and $\kappa$ from Eq:~\eqref{eq:thespl} as follows. On the one hand, the first process essentially views symmetrisation as a map from $A\otimes\cdots \otimes A$ into the symmetric subspace, where the symmetric subspace is thought of as a new kind of system rather than viewing it as living inside the original state space of $\mathcal{P}$. On the other hand, the second process can be viewed as embedding the symmetric subspace (viewed as an intrinsic system in its own right) into the original system. Since these embedding maps are the splitting of the idempotent, recall that they also satisfy: 
\beq
\InputIfFileExists{Diagrams/symKaroubi2.tikz}{}{\input{./figures/Diagrams/symKaroubi2.tikz}}\quad =\quad %
\begin{tikzpicture}
	\begin{pgfonlayer}{nodelayer}
		\node [style=none] (70) at (0.25, -2) {};
		\node [style=none] (71) at (0.25, 2) {};
		\node [style=right label] (72) at (0.25, -1.25) {$\mathsf{Sym}_A^n$};
	\end{pgfonlayer}
	\begin{pgfonlayer}{edgelayer}
		\draw [qWire, in=-90, out=90] (70.center) to (71.center);
	\end{pgfonlayer}
\end{tikzpicture}
}.
\eeq
 In light of all this, then, we view these embeddings as processes in \Ks which turn distinguishable systems into indistinguishable ones and vice-versa.  For example,  if we view the particles as originally being distinguishable by virtue of them being confined in separate spatial regions, then one may imagine the symmetrisation process as bringing the collection of particles together, or conversely, as pulling them apart again. 

This perspective naturally leads us to identify the  processes which represent, for example, introducing a third originally distinguishable particle to a system of two indistinguishable particles to create an indistinguishable triple. Conveniently, we can build such a process out of the processes we have already:
\beq
\InputIfFileExists{Diagrams/2plus1.tikz}{}{\input{./figures/Diagrams/2plus1.tikz}}\quad:=\quad %
\InputIfFileExists{Diagrams/2plus1to3.tikz}{}{\input{./figures/Diagrams/2plus1to3.tikz}}\qquad \text{and}\qquad %
\InputIfFileExists{Diagrams/2plus1dag.tikz}{}{\input{./figures/Diagrams/2plus1dag.tikz}}\quad:=\quad %
\InputIfFileExists{Diagrams/2plus1to3dag.tikz}{}{\input{./figures/Diagrams/2plus1to3dag.tikz}}.
\eeq
These diagrams satisfy nice conditions, such as:
\beq
\InputIfFileExists{Diagrams/2plus1eq1.tikz}{}{\input{./figures/Diagrams/2plus1eq1.tikz}} \quad = \quad %
\InputIfFileExists{Diagrams/sym3.tikz}{}{\input{./figures/Diagrams/sym3.tikz}}\qquad \text{and} \qquad %
\InputIfFileExists{Diagrams/2plus1eq2.tikz}{}{\input{./figures/Diagrams/2plus1eq2.tikz}}\quad =\quad %
\begin{tikzpicture}
	\begin{pgfonlayer}{nodelayer}
		\node [style=none] (70) at (0.25, -2) {};
		\node [style=none] (71) at (0.25, 2) {};
		\node [style=right label] (72) at (0.25, -1.25) {$\mathsf{Sym}_A^n$};
	\end{pgfonlayer}
	\begin{pgfonlayer}{edgelayer}
		\draw [qWire, in=-90, out=90] (70.center) to (71.center);
	\end{pgfonlayer}
\end{tikzpicture}
}\,,
\eeq
which follow from the fact that:
\beq
\InputIfFileExists{Diagrams/symComp.tikz}{}{\input{./figures/Diagrams/symComp.tikz}}\quad=\quad%
\InputIfFileExists{Diagrams/symKaroubi.tikz}{}{\input{./figures/Diagrams/symKaroubi.tikz}}.
\eeq


Now, following Eq.~\eqref{eq:KaroubiMapInt}, we can see how processes on distinguishable particles are mapped onto processes on indistinguishable particles, for example:
\beq
\InputIfFileExists{Diagrams/swapAA.tikz}{}{\input{./figures/Diagrams/swapAA.tikz}} \quad\mapsto\quad %
\InputIfFileExists{Diagrams/symSwap.tikz}{}{\input{./figures/Diagrams/symSwap.tikz}}\quad=\quad%
\InputIfFileExists{Diagrams/symSwap1.tikz}{}{\input{./figures/Diagrams/symSwap1.tikz}}\quad=\quad %
\begin{tikzpicture}
	\begin{pgfonlayer}{nodelayer}
		\node [style=none] (42) at (0, 2) {};
		\node [style=none] (43) at (0, -2) {};
		\node [style=right label] (44) at (0, -1.25) {$\mathsf{Sym}_A^2$};
	\end{pgfonlayer}
	\begin{pgfonlayer}{edgelayer}
		\draw [qWire, in=90, out=-90] (42.center) to (43.center);
	\end{pgfonlayer}
\end{tikzpicture}
},
\eeq
which follows from Eq.~\eqref{eq:permInvariance}. That is, swapping two indistinguishable particles leaves them invariant (at least on this operational level).

So far we have seen how to introduce symmeterised systems to a process theory. However, to really understand particles we need to go further than this. We need to see how these symmeterised systems split up into different types of particles. For example, in quantum theory we have that $\mathsf{Sym}_Q^2 = \mathsf{Fermion}\oplus \mathsf{Boson}$. To get there, we will first take a step back and formally introduce the $\oplus$ symbol in this generalised setting.

\subsection{Biproducts}

In the type of process theories we consider in this work (i.e., those with a well-defined notion of summation of processes), given two systems $A$ and $B$ one can define a new system, denoted $A\oplus B$, with the property that
processes from $A\oplus B$ to $C\oplus D$ are given by matrices of processes:
\beq
\left(\begin{array}{cc} f:A\to C & g:A\to D \\ h:B\to C & i:B\to D\end{array} \right).
\eeq
More generally, given two systems $\bigoplus_i A_i$ and $\bigoplus_j B_j$ we describe a process from $\bigoplus_i A_i$ to $\bigoplus_j B_j$  by the matrix
\beq
(f_{ij}:A_i\to B_j)_{ij}.
\eeq

Importantly, in these process theories one can define disjoint inclusion and projection maps:
\beq
\InputIfFileExists{Diagrams/biproduct1.tikz}{}{\input{./figures/Diagrams/biproduct1.tikz}}\qquad \text{and}\qquad %
\InputIfFileExists{Diagrams/biproduct2.tikz}{}{\input{./figures/Diagrams/biproduct2.tikz}}
\eeq
such that
\beq %
\InputIfFileExists{Diagrams/biproduct3.tikz}{}{\input{./figures/Diagrams/biproduct3.tikz}}\quad =\quad \delta_{ij} \ \  %
\begin{tikzpicture}
	\begin{pgfonlayer}{nodelayer}
		\node [style=none] (3) at (0.5, 1.25) {};
		\node [style=none] (4) at (0.5, -1) {};
		\node [style=right label] (5) at (0.5, -0.75) {$A_i$};
	\end{pgfonlayer}
	\begin{pgfonlayer}{edgelayer}
		\draw [qWire, in=90, out=-90] (3.center) to (4.center);
	\end{pgfonlayer}
\end{tikzpicture}
}\qquad \text{and}\qquad \sum_i\ \ %
\InputIfFileExists{Diagrams/biproduct5.tikz}{}{\input{./figures/Diagrams/biproduct5.tikz}}\quad =\quad %
\begin{tikzpicture}
	\begin{pgfonlayer}{nodelayer}
		\node [style=right label] (0) at (-0.25, -1) {$\bigoplus_jA_j$};
		\node [style=none] (3) at (-0.25, -1.5) {};
		\node [style=none] (4) at (-0.25, 1.5) {};
	\end{pgfonlayer}
	\begin{pgfonlayer}{edgelayer}
		\draw [qWire, in=-90, out=90] (3.center) to (4.center);
	\end{pgfonlayer}
\end{tikzpicture}
}.
\eeq
 In fact, the existence of such maps can be taken as the abstract definition of the biproduct of systems. \blk  These maps will be crucial when defining particle types in the next section. 

\subsection{Particle types}

Let us suppose that we have split the symmetrisation idempotents to obtain the systems $\mathsf{Sym}_A^n$ as defined in Sec.~\ref{sec:SplitSym}. Next we can express that idempotent as a sum of orthogonal idempotents $\mathsf{Part}_i$: 
\begin{align}\label{eq:symintoparti}
\InputIfFileExists{Diagrams/symKaroubi.tikz}{}{\input{./figures/Diagrams/symKaroubi.tikz}}\quad =  \quad\sum_i \ \ %
\InputIfFileExists{Diagrams/decomp5.tikz}{}{\input{./figures/Diagrams/decomp5.tikz}}.
\end{align}
There is a canonical choice for this as it can be shown that there is a unique ``most refined'' such decomposition, from which any other can be obtained via a coarse graining. This fact will be demonstrated in a forthcoming second version of this work. \blk

Note that $\mathsf{Part}_i$ unlike $\mathsf{Sym}$ cannot be thought of as operationally implementable processes as, for example, they do not preserve the normalisation of states.
 Instead, the collection $\{\mathsf{Part}_i\}_i$ can be thought of as describing a particular (nondestructive) measurement on the system, hence, any given $\mathsf{Part_i}$ describes what happens if the $i$th outcome of this measurement occurs. In fact, when this measurement is viewed as a measurement of the operationally symmetrised system $\mathsf{Sym}_A^n$ then it is in fact a non-disturbing measurement which we will view as the measurement of the particle type.

\blk


To see this more clearly, \blk we can split the symmetrisation operation through $\mathsf{Sym}_A^n$, and the idempotents $\mathsf{Part}_i$ through systems that we denote  $(S_A^n)_i$, hence:
\beq\label{eq:thesumbip}
\InputIfFileExists{Diagrams/symKaroubi1.tikz}{}{\input{./figures/Diagrams/symKaroubi1.tikz}} \quad=\quad \sum_i \ \ %
\InputIfFileExists{Diagrams/decomp4.tikz}{}{\input{./figures/Diagrams/decomp4.tikz}} \,.
\eeq

 This implies that \blk
 $\mathsf{Sym}_A^n = \bigoplus_i (S_A^n)_i$, and one can
\begin{quote}
\textit{interpret these new systems $(S_A^n)_i$ as the different particle types that arise within the theory}. 
\end{quote}
\blk

That is, a collection of $n$ systems of type $A$ when symmetrised will give different particle types labelled by $i$, and the system in the theory associated to the $n$ particles of type $i$ is $(S_A^n)_i$. 

For example, in quantum theory if we have a pair of qubits $\mathcal{B}(\mathds{C}^2)\otimes \mathcal{B}(\mathds{C}^2)$ and we symmetrise them, then we find that $\mathsf{Sym}_{\mathcal{B}(\mathds{C}^2)}^2 = \mathcal{B}(\mathds{C}^3)\oplus\mathcal{B}(\mathds{C}^1)$. In a second version of this work, we will see that this decomposition cannot be refined further and therefore is the unique, most refined decomposition. Hence, we have two sectors, the first is the fermionic sector where the pair of fermions can be viewed as a single qutrit system, and the boson sector where the pair of bosons can be viewed as a trivial one dimensional quantum system. 

Note, however, that this definition only really makes sense if the system $A\otimes...\otimes A$ cannot itself be (nontrivially) written as a direct sum decomposition. Classical theory is the key example where this is not the case, as any classical system can always be decomposed as a direct sum of one dimensional classical systems. In such cases we need to modify the definition such that particle types correspond only to any new decompositions which are created thanks to the symmetrisation. We, however, leave formalising this to a second version of this work.
\blk

\section{Discussion}

In this manuscript we presented two ways in which one may pursue the study of indistinguishable particles in general probabilistic theories. The first case corresponds to framing the questions in the traditional GPT framework. Here we associate the different types of indistinguishable particles with the different orbits of symmetric pure states that emerge under symmetry preserving transformations. In the case of bipartite quantum theory we recover the canonical two types of indistinguishable particles---bosons and fermions. 

In other theories, we find a difference when taking the pure states of our particle types to be symmetric extremal states or extremal symmetric states. In the Boxworld GPT, without imposing indistinguishability there already exists non-intersecting orbits of pure states, that could be taken as different particle types. Further imposing indistinguishability then does not result in further particle types. On the other hand, in the Spekkens toy theory we find three different particle types. 

The obvious next steps is to consider systems of multiple indistinguishable particles, and study the emergence of paraparticles in quantum theory.
There are also still various other GPTs that one may explore, such as the polygon GPTs~\cite{Janotta}, and one interesting question to pursue through them is to understand which properties of indistinguishable particles are non-classical, which may be unique to quantum theory, and how they relate to other features of the GPTs. Another interesting question pertains to how the type of emergent indistinguishable particles depends on the type of composition rule utilised by the GPT to compose their systems (two extreme cases of composition rules being the so-called min and max tensor products). Finally, one can ask what happens if we move on from tomographically-local GPTs: how do the so-called \textit{holistic} degrees of freedom that non-tomographically-local GPTs have impact the conceptualisation of indistinguishable particles in such GPTs? 

In the second half of this work, we presented the study of indistinguishability of particles in the diagrammatic formalism of process theories. Here the particle types emerged as the biproduct systems of which the symmetrisation idempotents can be expressed as a sum. We find that there is a unique most refined decomposition of the symmetrisation idempotent which we take as the canonical choice for defining particle types. Again, we recover the bosons and fermions for pairs of quantum systems.  

Interesting questions arise in this framework. For instance, one may ask how composition of systems may change or preserve the different particle types, and how this should be formalised. Secondly, when considering quantum theory, can we think of $\mathsf{Part}_F$ as wavefunction-antisymmetrisation and of $\mathsf{Part}_B$ as wavefunction-symmetrisation? Can this also be possible for other theories, such as Euclidean Jordan Algebras? And finally, do we recover the same particle types that the orbit-based approach yields?


\blk 
 
\section*{Acknowledgments}

JHS was supported by the National Science Centre, Poland (Opus project, Categorical Foundations of the Non-Classicality of Nature, project no. \\2021/41/B/ST2/03149).
ABS, VJW, and MK acknowledge support by the Foundation for Polish Science (IRAP project, ICTQT,contract no.~2018/MAB/5, co-financed by EU within Smart Growth Operational Programme).
This research was conducted while visiting the Okinawa Institute of Science and Technology (OIST) through the Theoretical Sciences Visiting Program (TSVP).

\bigskip

\bibliographystyle{unsrt}
\bibliography{indistbib}

\end{document}